\documentclass[a4paper,english,numberwithinsect]{lipics-v2018}

\usepackage{latexsym,color,graphicx,amsmath,amssymb,subcaption}
\usepackage{url}

\def\reals{{\mathbb R}}
\def\eps{{\varepsilon}}

\def\dd{{\sf dist}}

\def\fd{{\sf fd}}

\def\ww{{\bf w}}
\def\xx{{\bf x}}
\def\t{{\bf t}}
\def\ss{{\bf s}}
\def\gg{{\bf g}}
\def\mm{{\bf w}}
\def\xxi{{\bf \xi}}

\def\ff{{\bf f}}
\def\g{{\bf g}}
\def\s{{\bf s}}
\def\dist{{\sf dist}}

\nolinenumbers

\title{General techniques for approximate incidences
and their application to the camera posing problem}

\titlerunning{Approximation Algorithms for Camera Posing}
\author{Dror Aiger}%
{Google, aigerd@google.com} {}%
{}%
{}
\author{Haim Kaplan}%
{School of Computer Science, Tel Aviv University, Tel~Aviv, and Google}%
{haimk@tau.ac.il}%
{}%
{Partially supported by ISF grant 1841/14, by grant 1367/2016
   from the German-Israeli Science Foundation (GIF), and by Blavatnik Research Fund in Computer Science at Tel Aviv University.}%
\author{Efi Kokiopoulou}%
{Google, efi@google.com} {}%
{}%
{}
\author{Micha Sharir}%
{School of Computer Science, Tel Aviv University, Tel~Aviv,
   Israel}%
{michas@tau.ac.il}%
{}%
{Partially supported by ISF Grant 260/18,  by grant 1367/2016
   from the German-Israeli Science Foundation (GIF), and by
   Blavatnik Research Fund in Computer Science at Tel Aviv University.}
\author{Bernhard Zeisl}%
{Google, bzeisl@google.com}%
{}%
{}%
{}
\authorrunning{D. Aiger, H. Kaplan, E. Kokiopoulou, M. Sharir, B. Zeisl}

\Copyright{D. Aiger, H. Kaplan, E. Kokiopoulou, M. Sharir, B. Zeisl}


\pdfoutput=1

\EventEditors{}
\EventNoEds{2}
\EventLongTitle{}
\EventShortTitle{}
\EventAcronym{}
\EventYear{}
\EventDate{}
\EventLocation{}
\EventLogo{}
\SeriesVolume{}
\ArticleNo{} % <- Your article number goes here
\begin{document}

\maketitle

\begin{abstract}
We consider the classical camera pose estimation problem that arises in many 
computer vision applications, in which we are given $n$ 2D-3D correspondences 
between points in the scene and points in the camera image (some of which are 
incorrect associations), and where we aim to determine the camera pose 
(the position and orientation of the camera in the scene) from this data.
We demonstrate that this posing problem can be reduced to the problem of 
computing $\eps$-approximate incidences between two-dimensional surfaces 
(derived from the input correspondences) and points (on a grid) in a four-dimensional pose space.
Similar reductions can be applied to other camera pose problems, as well as
to similar problems in related application areas.

We describe and analyze three techniques for solving the resulting
$\eps$-approximate incidences problem in the context of our camera posing application.
The first is a straightforward assignment of surfaces to the cells of a grid
(of side-length $\eps$) that they intersect.
The second is a variant of a primal-dual technique, recently introduced by
a subset of the authors~\cite{aiger2017} for different (and simpler)
applications. The third is a non-trivial generalization of a data structure
Fonseca and Mount~\cite{fonseca2010}, originally designed for the case of hyperplanes. 
% The data structure approximates
% the surfaces by a set of canonical surfaces whose number depends only on $\eps$.
% It then recurses on the problem resolution, and at each step re-canonizes the
% relevant surfaces into a smaller number of coarser approximations.
% This method is particularly appealing for a large number of input surfaces, such that many of them get rounded to the same canonical surface during the recursion.
We present and analyze this technique in full generality, and then apply it to the camera posing problem at hand.

We compare our methods experimentally on real and synthetic data.
Our experiments show that for the typical values of $n$ and $\eps$,
the primal-dual method is the fastest, also in practice.
\end{abstract}

\section{Introduction} \label{sec:intro}
Camera pose estimation is a fundamental problem in computer vision, which aims
at determining the pose and orientation of a camera solely from an image.
This localization problem appears in many interesting real-world applications,
such as for the navigation of self-driving cars~\cite{haene2017},
in incremental environment mapping such as Structure-from-Motion (SfM)~\cite{agarwal2009,pollefeys2004,schonberger2016},
or for augmented reality~\cite{klein2007,middelberg2014,sweeney2015},
where a significant component are algorithms that aim to estimate
an accurate camera pose in the world from image data.

Given a three-dimensional point-cloud model of a scene, the classical, but also
state-of-the-art approach to absolute camera pose estimation consists of a two-step procedure.
First, one matches a large number of features in the two-dimensional camera image
with corresponding features in the three-dimensional scene.
Then one uses these putative correspondences to determine the pose and orientation of the camera.
Typically, the matches obtained in the first step contain many incorrect associations, forcing the second step
to use filtering techniques to reject incorrect matches.
Subsequently, the absolute 6 degrees-of-freedom (DoF) camera pose is estimated,
for example, with a perspective $n$-point pose solver~\cite{haralick1994} within a RANSAC scheme~\cite{fischler1981}.

In this work we concentrate on the second step of the camera pose problem. That is, we consider
the task of estimating the camera pose and orientation from a (potentially large) set of $n$
already calculated image-to-scene correspondences.

Further, we assume that we are given
% the camera is calibrated -- i.e., its intrinsic parameters are fixed -- and that 
a common direction between the world and camera frames. For example, inertial sensors, 
available on any smart-phone nowadays, allow to estimate the vertical gravity direction 
in the three-dimensional camera coordinate system. This
alignment of the vertical direction fixes two degrees of freedom for the rotation 
between the frames and we are left to estimate four degrees of freedom out of the general six.
To obtain four equations (in the four remaining degrees of freedom), this setup requires two 
pairs of image-to-scene correspondences\footnote{%
  As we will see later in detail, each correspondence imposes two constraints on the camera pose.} 
for a minimal solver. Hence a corresponding naive RANSAC-based scheme requires $O(n^2)$ 
filtering steps, where in each iterations a pose hypothesis based on a different pair 
of correspondences is computed and verified against all other correspondences.

Recently, Zeisl et al.~\cite{zeisl2015} proposed a Hough-voting inspired
outlier filtering and camera posing approach, which computes the camera pose
up to an accuracy of $\eps>0$ from a set of $2$D-$3$D correspondences, in $O(n/\eps^2)$ time, 
under the same alignment assumptions of the vertical direction.
In this paper we propose new algorithms that work considerably faster in practice, but under milder assumptions.
Our method is based on a reduction of the problem to a problem of counting
\emph{$\eps$-approximate incidences} between points and surfaces, where a point
$p$ is $\eps$-approximately incident (or just $\eps$-incident) to a surface $\sigma$
if the (suitably defined) distance between $p$ and $\sigma$ is at most $\eps$.
This notion has recently been introduced by a subset of the authors 
in~\cite{aiger2017}, and applied in a variety of instances, involving 
somewhat simpler scenarios than the one considered here.
Our approach enables us to compute a camera pose when the number of correspondences $n$ is large, and many of which are expected to be outliers.
In contrast, a direct application of RANSAC-based methods on such inputs is very slow, since the fraction of inliers is small.
%It takes many trials to find a pair of good correspondences if drawn at random, since the fraction of inliers is small.
In the limit, trying all pairs of matches involves $\Omega(n^2)$ RANSAC iterations.
Moreover, our methods enhance the quality of the posing considerably~\cite{zeisl2015}, 
since each generated candidate pose is close to (i.e., consistent with) with many of the correspondences.

\medskip
\noindent
{\bf Our results.}
We formalize the four degree-of-freedom camera pose problem as an approximate incidences problem in Section \ref{sec:spatial_pos}.
Each 2D-3D correspondence is represented as a two-dimensional surface in the 4-dimensional pose-space, which
is the locus of all possible positions and orientations of the camera that fit the correspondence exactly.
Ideally, we would like to find a point (a pose) that lies on as many surfaces as possible, but since
we expect the data to be noisy, and the exact problem is inefficient to solve anyway, we settle for 
an approximate version, in which we seek a point with a large number of approximate incidences with the surfaces. 

Formally, we solve the following problem. We have an error parameter $\eps>0$, we lay down a grid on
$[0,1]^d$ of side length $\eps$, and compute, for each vertex $v$ of the grid, a count $I(v)$ of
surfaces that are approximately incident to $v$, so that
(i) every surface that is $\eps$-incident to $v$ is counted in $I(v)$, and 
(ii) every surface that is counted in $I(v)$ is $\alpha\eps$-incident to $v$, for some small constant $\alpha>1$
(but not all $\alpha\eps$-incident surfaces are necessarily counted). We output the grid vertex $v$ with the
largest count $I(v)$ (or a list of vertices with the highest counts, if so desired).

As we will comment later, (a) restricting the algorithm to grid vertices only does not miss a good pose $v$:
a vertex of the grid cell containing $v$ serves as a good substitute for $v$, and (b) we have no real
control on the value of $I(v)$, which might be much larger than the number of surfaces that are $\eps$-incident
to $v$, but all the surfaces that we count are `good'---they are reasonably close to $v$. In the computer vision
application, and in many related applications, neither of these issues is significant.

We give three algorithms for this camera-pose approximate-incidences problem.
The first algorithm simply computes the grid cells
that each surface intersects, and considers the number of intersecting surfaces per cell as its approximate $\eps$-incidences count.
This method takes time $O\left(\frac{n}{\eps^2}\right)$ for all vertices of our $\eps$-size grid.
We then describe a faster algorithm using geometric duality, in Section \ref{sec:primal_dual}. 
It uses a coarser grid in the primal space and switches to a dual 5-dimensional space 
(a 5-tuple is needed to specify a $2$D-$3$D correspondence and its surface, now dualized to a point).
In the dual space each query (i.e., a vertex of the grid) becomes a 3-dimensional surface,
and each original 2-dimensional surface in the primal 4-dimensional space becomes a point.
This algorithm takes $O\left( \frac{n^{3/5}}{\eps^{14/5}} + n + \frac{1}{\eps^4} \right)$ time, 
and is asymptotically faster than the simple algorithm for $n > 1/\eps^2$.

Finally, we give a general method for constructing an approximate incidences data structure for general
$k$-dimensional algebraic surfaces (that satisfy certain mild conditions) in $\reals^d$,
in Section \ref{sec:fonseca}. It extends the technique of Fonseca and Mount~\cite{fonseca2010},
designed for the case of hyperplanes, and takes $O(n + {\rm poly}(1/\eps))$ time, where the degree of the polynomial
in $1/\eps$ depends on the number of parameters needed to specify a surface, the dimension of the surfaces, 
and the dimension of the ambient space. We first present and analyze this technique in full generality, and then apply it to
the surfaces obtained for our camera posing problem. In this case, the data structure requires
$O(n+ 1/\eps^6)$ storage and is constructed in roughly the same time.
This is asymptotically faster than our primal-dual scheme when $n\ge 1/\eps^{16/3}$ 
(for $n\ge 1/\eps^7$ the $O(n)$ term dominates and these two methods are asymptotically the same).
Due to its generality, the latter technique is easily adapted to other surfaces and thus is of general interest and potential.
In contrast, the primal-dual method requires nontrivial adaptation as it switches from one 
approximate-incidences problem to another and the dual space and its distance function 
depend on the type of the input surfaces.

We implemented our algorithms and compared their performance on real and synthetic data.
Our experimentation shows that, for commonly used values of $n$ and $\eps$ in 
practical scenarios ($n \in [8K, 32K]$, $\eps \in [0.02, 0.03]$),
the primal-dual scheme is considerably faster than the other algorithms, and should thus be the method of choice. 
Due to lack of space, the experimentation details are omitted in this version, with the exception of a few highlights.
They can be found in the appendix.
% These results are presented in Section~\ref{sec:experiments}.

\section{From camera positioning to approximate incidences} \label{sec:spatial_pos}

Suppose we are given a pre-computed three-dimensional scene and a two-dimensional
picture of it. Our goal is to deduce from this image the location and orientation
of the camera in the scene. In general, the camera, as a rigid body in 3-space,
has six degrees of freedom, three of translation and three of rotation (commonly
referred to as the \emph{yaw, pitch and roll}).
We simplify the problem by making the realistic assumption, that the vertical direction 
of the scene is known in the camera coordinate frame (e.g., estimated by en inertial sensor on smart phones).
This allows us to rotate the camera coordinate frame such that its $z$-axis is parallel to the world $z$-axis,
thereby fixing the pitch and roll of the camera and leaving only four degrees of freedom
$(x,y,z,\theta)$, where $c=(x,y,z)$ is the location of the camera center, say,
and $\theta$ is its yaw, i.e.\ horizontal the orientation of the optical axis around the vertical direction.
See Figure \ref{fig:yaw-pitch}.

\begin{figure}[h]
\centering
\begin{minipage}{0.9\linewidth}
    \includegraphics[width=\linewidth]{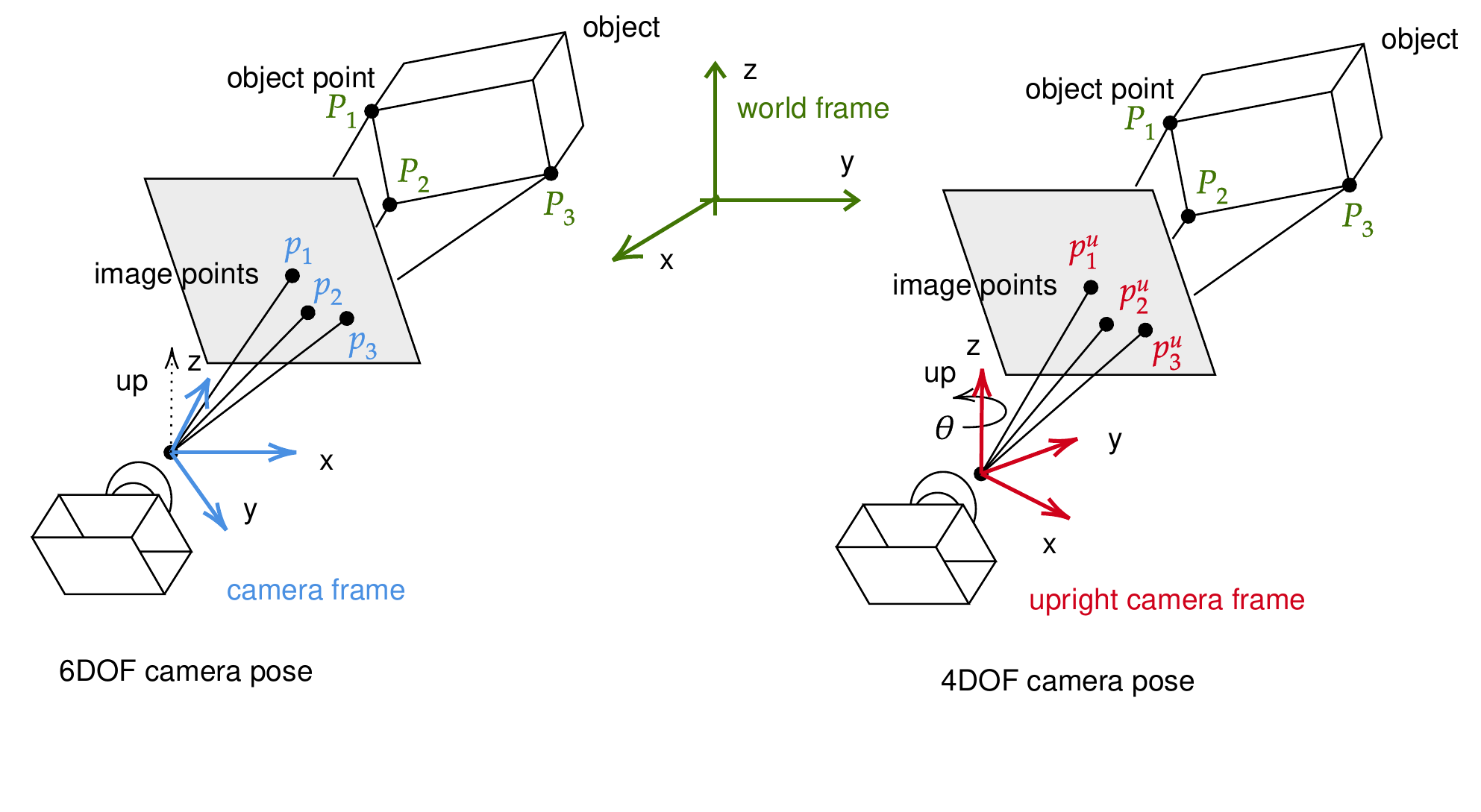}
    \caption{With the knowledge of a common vertical direction
    between the camera and world frame the general 6DoF camera posing problem reduces
    to estimating 4 parameters.
    This is the setup we consider in our work. 
    }
    \label{fig:yaw-pitch}
\end{minipage}
\end{figure}

By preprocessing the scene, we record the spatial coordinates $w=(w_1,w_2,w_3)$ of a discrete (large) set of salient
points. We assume that some (ideally a large number) of the distinguished points are identified in the camera image,
resulting in a set of image-to-scene correspondences.
Each correspondence $\mm = \{w_1, w_2, w_3, \xi, \eta\}$ is parameterized by five parameters,
the spatial position $w$ and the position $v=(\xi,\eta)$ in the camera plane of view of the same salient point.
Our goal is to find a camera pose $(x,y,z,\theta)$ so that as many correspondences
as possible are (approximately) \emph{consistent} with it, i.e., the ray from the camera 
center $c$ to $w$ goes approximately through $(\xi, \eta)$ in the image plane, when the yaw of the camera is $\theta$.

\subsection{Camera posing as an $\eps$-incidences problem}

% In the following we further develop and formalize the aforementioned  connection between the camera posing problem
% and the problem of computing approximate incidences, where the latter is a variant of the
% problems considered in the earlier work~\cite{aiger2017}.
Each correspondence and its 5-tuple $\mm$ define a two-dimensional surface $\sigma_\mm$
in parametric 4-space, which is the locus of all poses $(x,y,z,\theta)$ of the camera
at which it sees $w$ at coordinates $(\xi,\eta)$ in its image.
For $n$ correspondences, we have a set of $n$ such surfaces.
We prove that each point in the parametric $4$-space of camera poses that is
close to a surface $\sigma_\mm$, in a suitable metric defined in that $4$-space, represents a camera pose where $w$ is
projected to a point in the camera viewing plane that is close to $(\xi,\eta)$, and vice versa
(see Section~\ref{sec:sigma} for the actual expressions for these projections).
Therefore, a point in $4$-space that is close to a large number of surfaces
represents a camera pose with many approximately consistent correspondences, which is a
strong indication of being close to the correct pose.

Extending the notation used in the earlier work~\cite{aiger2017},
we say that a point $q$ is \emph{$\eps$-incident} to a surface $\sigma$ if $\dd(q,\sigma) \le \eps$.
Our algorithms approximate, for each vertex of a grid $G^\eps$ of side length $\eps$, the number of
$\eps$-incident surfaces and suggest the vertex with the largest count as the best candidate for the camera pose.
% In~\citep{aiger2017} we introduced efficient algorithms to compute approximate
% incidences between a set of points in two or three dimensions and a set of certain types
% of curves or surfaces, such as lines or circles in the plane, and lines, circles, planes or spheres in $3$-space. 
This work extends the approximate incidences methodology in~\cite{aiger2017} 
to the (considerably more involved) case at hand.

\subsection{The surfaces $\sigma_\mm$} \label{sec:sigma}

Let $w=(w_1,w_2,w_3)$ be a salient point in $\reals^3$, and assume that the camera
is positioned at $(c,\theta) = (x,y,z,\theta)$. We represent the orientation of the vector $w-c$,
within the world frame, by its spherical coordinates
$(\varphi,\psi)$, except that, unlike the standard convention, we take $\psi$ to be the
angle with the $xy$-plane (rather than with the $z$-axis):
\[
\tan\psi  = \frac{w_3-z}{\sqrt{(w_1-x)^2+(w_2-y)^2}} \quad\quad\quad
\tan\varphi  = \frac{w_2-y}{w_1-x}
\]
In the two-dimensional frame of the camera the $(\xi,\eta)$-coordinates model the \emph{view} of $w$, which differs from above polar representation of the vector $w-c$ only by the polar orientation $\theta$ of the viewing plane itself. Writing $\kappa$ for $\tan\theta$, we have
\begin{equation} 
\label{cweqs}
\xi = \tan(\varphi-\theta) = \frac{\tan\varphi - \tan\theta}{1 + \tan\varphi \tan\theta} 
= \frac{(w_2-y) - \kappa (w_1-x)}{(w_1-x) + \kappa (w_2-y)} ,
\end{equation}
\[
\eta = \tan\psi = \frac{w_3-z}{\sqrt{(w_1-x)^2+(w_2-y)^2}} .
\]
We note that using $\tan\theta$ does not distinguish between $\theta$ and $\theta+\pi$,
but we will restrict $\theta$ to lie in $[-\pi/4,\pi/4]$ or in similar narrower
ranges, thereby resolving this issue.

We use $\reals^4$ with coordinates $(x,y,z,\kappa)$ as our primal space, where each point models
a possible pose of the camera. Each correspondence $\ww$ is parameterized by the triple
$(w,\xi,\eta)$, and defines a two-dimensional algebraic surface $\sigma_\ww$ of degree at most $4$, 
whose equations (in $x,y,z,\kappa$)
are given in \eqref{cweqs}. It is the locus of all camera poses $v=(x,y,z,\kappa)$
at which it sees $w$ at image coordinates $(\xi,\eta)$. 
% The surface $\sigma_\ww$
% is two-dimensional, being the intersection of the
% two quadratic hypersurfaces defined by the equations in \eqref{cweqs}, and is therefore of degree $4$.
We can rewrite these equations into the following parametric representation of $\sigma_\ww$,
expressing $z$ and $\kappa$ as functions of $x$ and $y$:
\begin{equation} \label{cparw}
\kappa  = \frac{(w_2-y) - \xi (w_1-x)} {(w_1-x) + \xi (w_2-y)}  \quad\quad\quad
z  = w_3 - \eta \sqrt{(w_1-x)^2+(w_2-y)^2} .
\end{equation}
% We use \eqref{cparw} rather than \eqref{cweqs} as the parametric equations defining the surface $\sigma_\ww$.
% We use the following notation where, 
For a camera pose $v=(x,y,z,\kappa)$, and a point $w=(w_1,w_2,w_3)$, we write
\begin{equation} \label{eq:fg}
F(v;w)  =
\frac{(w_2-y) - \kappa (w_1-x)} {(w_1-x) + \kappa (w_2-y)}  \quad\quad\quad
G(v;w)  =
\frac{w_3-z}{\sqrt{(w_1-x)^2+(w_2-y)^2}}.
\end{equation}

In this notation we can write the Equations \eqref{cweqs} characterizing $\sigma_\ww$
(when regarded as equations in $v$) as
$\xi  = F(v;w)$ and $\eta  = G(v;w)$.

\subsection{Measuring proximity}

Given a guessed pose $v=(x,y,z,\kappa)$ of the camera, we want to measure how
well it fits the scene that the camera sees. For this, given a correspondence
$\ww=(w,\xi,\eta)$, we define the \emph{frame distance} $\fd$
between $v$ and $\ww$ as the $L_\infty$-distance between $(\xi,\eta)$ and
$(\xi_v,\eta_v)$, where, as in Eq.~\eqref{eq:fg},
$ \xi_v  = F(v;w)$,
$\eta_v  = G(v;w)$.
That is,
\begin{equation} \label{eq:fd}
\fd(v,\ww) = \max \left\{ |\xi_v-\xi|,\; |\eta_v-\eta| \right\} .
\end{equation}
Note that $(\xi_v,\eta_v)$ are the coordinates at which the camera would see $w$ if
it were placed at position $v$, so the frame distance is the $L_\infty$-distance between
these coordinates and the actual coordinates $(\xi,\eta)$ at which the camera sees $w$;
this serves as a natural measure of how close $v$ is to the actual pose of the camera.

We are given a viewed scene of $n$ distinguished points (correspondences)
$\ww=(w,\xi,\eta)$. Let $S$ denote the set of $n$ surfaces $\sigma_\ww$,
representing these correspondences.
We assume that the salient features $w$ and the camera are all located within some bounded
region, say $[0,1]^3$.
The replacement of $\theta$ by $\kappa=\tan\theta$ makes
its range unbounded, so we break the problem into four subproblems, in each of which $\theta$
is confined to some sector. In the first subproblem we assume that $-\pi/4\le\theta\le\pi/4$,
so $-1\le\kappa\le 1$. The other three subproblems involve the ranges
$[\pi/4,3\pi/4]$, $[3\pi/4,5\pi/4]$, and $[5\pi/4,7\pi/4]$. We only consider here the
first subproblem; the treatment of the others is fully analogous. In each such
range, replacing $\theta$ by $\tan\theta$ does not incur the ambiguity of
identifying $\theta$ with $\theta+\pi$.

Given an error parameter $\eps>0$, we seek an approximate pose $v$ of the camera,
at which many correspondences  $\ww$ are within frame distance at most $\eps$ from $v$,
as given in (\ref{eq:fd}).

The following two lemmas relate our frame distance to the Euclidean distance. Their (rather technical) proofs are given in the appendix.

%%%%%%%%%%%%%%%%%%%%%%%%%%%%
\begin{lemma} \label{xetau}
Let $v=(x,y,z,\kappa)$, and let $\sigma_\ww$ be the surface associated with a correspondence
$\ww=\{w_1,w_2,w_3,\xi,\eta \}$. Let $v'$ be a point on $\sigma_\ww$ such that $|v-v'| \le \eps$ (where $|\cdot |$ denotes the
Euclidean norm).
If

\noindent
(i) $\left| (w_1-x) + \kappa(w_2-y)\right| \ge a>0$, and

\noindent
(ii) $(w_1-x)^2 + (w_2-y)^2 \ge a>0$, for some absolute constant $a$,

\noindent
then $\fd(v,\ww) \le\beta\eps$ for some constant $\beta$ that depends on $a$.
\end{lemma}
%%%%%%%%%%%%%%%%%%%%%%%%%%%%

Informally, Condition (i) requires that
the absolute value of the $\xi = \tan (\varphi-\theta)$ 
coordinate of the position of $w$ in the viewing plane,
with the camera positioned at $v$, is not too large (i.e., that
$|(\varphi-\theta) |$ is not too close to $\pi/2$). We can ensure this property
by restricting the camera image to some suitably bounded $\xi$-range.

Similarly, Condition (ii) requires that the $xy$-projection of the vector $w-c$
is not too small. It can be violated in two scenarios. Either we look
at a data point that is too close to $c$, or we see it looking too much `upwards' or `downwards'.
We can ensure that the latter situation does not arise, by restricting the camera image, as in the preceding paragraph,
to some suitably bounded $\eta$-range too. That done, we ensure that the former situation does not arise by
requiring that the physical distance between $c$ and $w$ be at least some multiple of $a$.

The next lemma establishes the converse connection.

%%%%%%%%%%%%%%%%%%%%%%%
\begin{lemma} \label{reverse-primal}
Let $v=(x,y,z,\kappa)$ be a camera pose and $\ww=\{w_1,w_2,w_3,\xi,\eta \}$
a correspondence, such that $\fd(v,\ww) \le\eps$. Assume that
$\left| (w_1-x) + \xi(w_2-y)\right| \ge a>0$, for some absolute constant $a$,
and consider the point $v'=(x,y,z',\kappa')\in \sigma_\ww$ where (see Eq.~\eqref{cparw})
\[
z' = w_3 - \eta \sqrt{(w_1-x)^2+(w_2-y)^2}   \;\;\;\;\;\;\;\;\;\;
\kappa'  = \frac{(w_2-y) - \xi (w_1-x)} {(w_1-x) + \xi (w_2-y)} .
\]
Then $|z-z'| \le \sqrt{2}\eps$ and
$|\kappa - \kappa'| \le c\eps$, for some constant $c$, again depending on $a$.
\end{lemma}
%%%%%%%%%%%%%%%%%%%%%%%%%%%%
Informally, the condition $\left| (w_1-x) + \xi(w_2-y)\right| \ge a>0$ means that
the orientation
of the camera, when it is positioned at $(x,y)$ and  sees $w$ at coordinate $\xi$
of the viewing plane is not
too close to $\pm \pi/2$.
 This is a somewhat artificial constraint that is satisfied by our restriction on the allowed yaws of the camera (the range of $\kappa$).

%\noindent
\medskip
\noindent{\bf A Simple algorithm.}
Using Lemma \ref{reverse-primal} and Lemma \ref{xetau} we can derive a
simple naive solution which does not require any of the sophisticated machinery developed in this work.
We construct a grid $G$ over $Q=[0,1]^3\times[-1,1]$, of cells $\tau$, each of
dimensions $\eps \times \eps \times 2\sqrt{2}\eps \times 2a\eps$,
where $a$ is the constant of Lemma \ref{reverse-primal}. We use this non-square grid $G$ since we want to find $\eps$-approximate incidences in
terms of frame distance.
For each  cell $\tau$ of $G$ we compute the number of surfaces $\sigma_\ww$ that intersect $\tau$.
This gives an approximate incidences count for the center of $\tau$.
Further details and a precise statement can be found in the appendix.

%%%%%%%%%%%%%%%%%%%%%%%%%%%%%%%%%%%%%%%%%%%%%%%%
\section{Primal-dual algorithm for geometric proximity} \label{sec:primal_dual}

Following the general approach in \cite{aiger2017}, we use a suitable duality,
with some care. We write $\eps = 2\gamma\delta_1\delta_2$, for suitable parameters
$\gamma$, and $\eps/(2\gamma)\le \delta_1,\,\delta_2 \le 1$, whose concrete values are
fixed later, and apply the decomposition scheme developed in \cite{aiger2017}
tailored to the case at hand. Specifically, we consider the coarser grid $G_{\delta_1}$
in the primal space, of cell dimensions
$\delta_1 \times \delta_1 \times \sqrt{2}\delta_1 \times c\delta_1$, where $c$ is 
is the constant from Lemma~\ref{reverse-primal}, that tiles up the domain $Q=[0,1]^3\times[-1,1]$  of possible camera positions.
For each cell $\tau$ of $G_{\delta_1}$, let $S_\tau$ denote the set of surfaces that cross
either $\tau$ or one of the eight cells adjacent to $\tau$ in the $(z,\kappa)$-directions.\footnote{%
  The choice of $z$, $\kappa$ is arbitrary, but it is natural for the analysis, given
  in the appendix.}
The duality is illustrated in Figure \ref{fig:primal_dual}.

\begin{figure}[h]
\centering
\includegraphics[width=0.6\linewidth]{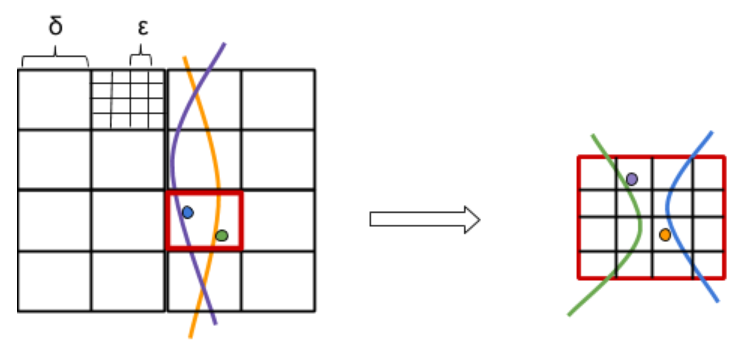}
\caption{A schematic illustration of our duality-based algorithm.}
\label{fig:primal_dual}
\end{figure}

We discretize the set of all possible positions of the camera by
the vertices of the finer grid $G_\eps$, defined as $G_{\delta_1}$, with $\eps$ replacing $\delta_1$,  that tiles up $Q$.
The number of these candidate positions is $m:=O(1/\eps^4)$. For each
vertex  $q\in G_\eps$, we want to approximate the number of surfaces that are
$\eps$-incident to $q$, and output the vertex with the largest count as the best candidate for the position of the camera.
% (Note that if $p$ is the real position of the camera then each surface that is $\eps$-close to $p$
% is counted as an $\eps$-incidence of one of the vertices of the cell of $G_\eps$ containing $p$.)
%
Let $V_\tau$ be the subset of $G_\eps$ contained in $\tau$. We ensure that the boxes of
$G_{\delta_1}$ are pairwise disjoint by making them half open, in the sense that
if $(x_0,y_0,z_0,\kappa_0)$ is the vertex of a box that has the smallest coordinates,
then the box is defined by
$x_0  \le x < x_0 + \delta_1$,
$y_0  \le y < y_0 + \delta_1$,
$z_0  \le z < z_0 + \sqrt{2}\delta_1$,
$\kappa_0  \le \kappa < \kappa_0 + c\delta_1$.
This makes the sets $V_\tau$ pairwise disjoint as well.
Put $m_\tau = |V_\tau|$ and $n_\tau = |S_\tau|$. We have
$m_\tau = O\left( (\delta_1/\eps)^4 \right)$ for each $\tau$.
Since the surfaces $\sigma_\ww$ are two-dimensional algebraic surfaces of constant degree, each of them crosses $O(1/\delta_1^2)$ cells of $G_{\delta_1}$, so we have
$\sum_\tau n_\tau = O(n/\delta_1^2)$.

We now pass to the dual five-dimensional space. Each point in that space
represents a correspondence $\ww=(w_1,w_2,w_3,\xi,\eta)$. We use the first
three components $(w_1,w_2,w_3)$ as the first three coordinates, but
modify the $\xi$- and $\eta$-coordinates in a manner that depends on
the primal cell $\tau$.
Let $c_\tau=(x_\tau,y_\tau,z_\tau,\kappa_\tau)$ be the midpoint of the primal box
$\tau$. For each $\sigma_\ww\in S_\tau$ we map $\ww=(w,\xi,\eta)$, where
$w=(w_1,w_2,w_3)$, to the point $\ww_\tau = (w_1,w_2,w_3,\xi_\tau,\eta_\tau)$, where
$\xi_\tau  = \xi - F(c_\tau;w)$ and
$\eta_\tau  = \eta - G(c_\tau;w)$,
with $F$ and $G$ as given in \eqref{eq:fg}.
We have
%%%%%%%%%%%%%%%%%%%%%%%%%%%%
\begin{corollary} \label{xetau-corr}
If $\sigma_\ww$ crosses $\tau$ then
$|\xi_\tau|,\;|\eta_\tau| \le\gamma\delta_1$, for some absolute constant $\gamma$,
provided that the following two properties hold, for some absolute constant $a>0$ (the constant $\gamma$ depends on $a$).

\noindent
(i) $\left| (w_1-x_\tau) + \kappa_\tau(w_2-y_\tau)\right| \ge a$, and

\noindent
(ii) $(w_1-x_\tau)^2 + (w_2-y_\tau)^2 \ge a$,
where $(x_\tau,y_\tau)$ are the $(x,y)$-coordinates of the center of $\tau$.
\end{corollary}
%\noindent{\bf Proof.}
\begin{proof}
If $\sigma_\ww \in S_\tau$ then it contains a point $v'$ such that $|v'-c_\tau| \le c'\delta_1$,
for a suitable absolute constant $c'$ (that depends on $c$). We now apply Lemma \ref{xetau},
recalling (\ref{eq:fd}).
\end{proof}
%%%%%%%%%%%%%%%%%%%%%%%%%%%%

\medskip
We take the $\gamma$ provided by Corollary \ref{xetau-corr} as the $\gamma$ in the definition
of $\delta_1$ and $\delta_2$.
We map each point $v\in V_\tau$ to the dual surface
$\sigma^*_v = \sigma^*_{v;\tau} = \{\ww_\tau \mid v\in\sigma_\ww\}$.
Using \eqref{eq:fg}, we have
\[
\sigma^*_{v;\tau} = \{ ( w,\; F(v;w) - F(c_\tau;w), G(v;w) - G(c_\tau;w) ) \mid w=(w_1,w_2,w_2)\in[0,1]^3 \} .
\]
By Corollary~\ref{xetau-corr}, the points $\ww_\tau$, for the surfaces $\sigma_\ww$ that cross
$\tau$, lie in the region $R_\tau = [0,1]^3\times [-\gamma\delta_1,\gamma\delta_1]^2$.
We partition $R_\tau$ into a grid $G_{\delta_2}$ of $1/\delta_2^5$ small congruent boxes, each of dimensions
$
\delta_2 \times \delta_2 \times \delta_2 \times (2\gamma\delta_1\delta_2) \times (2\gamma\delta_1\delta_2) =
\delta_2 \times \delta_2 \times \delta_2 \times \eps \times \eps$.

Exactly as in the primal setup, we make each of these boxes half-open, 
% so that if the vertex
% with smallest coordinates of a box is $(w_1^0,w_2^0,w_3^0,\xi_\tau^0,\eta_\tau^0)$,
% then the box is given by
% $w_1^0  \le w_1 < w_1^0+\delta_2$,
% $w_2^0  \le w_2 < w_2^0+\delta_2$,
% $w_3^0  \le w_3 < w_3^0+\delta_2$,
% $\xi_\tau^0  \le \xi_\tau < \xi_\tau^0+\eps$,
% $\eta_\tau^0  \le \eta_\tau < \eta_\tau^0+\eps$,
thereby making the sets of dual vertices in the smaller boxes pairwise disjoint.
We assign to each of these dual cells $\tau^*$ the set $S^*_{\tau^*}$ of dual points
that lie in $\tau^*$, and the set $V^*_{\tau^*}$ of the dual surfaces that cross either
$\tau^*$ or one of the eight cells adjacent to $\tau^*$ in the $(\xi_\tau, \eta_\tau)$-directions.
%
%\looseness=-1
Put $n_{\tau^*} = |S^*_{\tau^*}|$ and $m_{\tau^*} = |V^*_{\tau^*}|$.
Since the dual cells are pairwise disjoint, we have $\sum_{\tau^*} n_{\tau^*} = n_\tau$.
Since the dual surfaces are three-dimensional algebraic surfaces of constant degree, each of them crosses $O(1/\delta_2^3)$ grid cells, so
$\sum_{\tau^*} m_{\tau^*} = O\left( m_\tau / \delta_2^3 \right)$.

We compute, for each dual surface $\sigma^*_v$, the sum $\sum_{\tau^*} |S^*_{\tau^*}|$,
over the dual cells $\tau^*$ that are either crossed by $\sigma^*_v$ or that one of
their adjacent cells in the $(\xi_\tau, \eta_\tau)$-directions is crossed by $\sigma^*_v$.
We output the vertex $v$ of $G_\eps$ with the largest resulting count, over all primal cells $\tau$.

The following theorem  establishes the correctness of our technique. Its proof is given in Appendix B.

%%%%%%%%%%%%%%%%%%%%%%%
\begin{theorem} \label{correct}
Suppose that for every cell $\tau\in G_{\delta_1}$ and for every point
$v=(x,y,z,\kappa) \in V_\tau$ and every $\ww = ((w_1,w_2,w_3),\xi,\eta)$ such that
$\sigma_\ww$ intersects either $\tau$ or one of its adjacent cells in the
$(\xi_\tau, \eta_\tau)$-directions, we have that, for some absolute constant $a > 0$,

\noindent
(i) $\left| (w_1-x) + \kappa(w_2-y)\right| \ge a$,

\noindent
(ii) $(w_1-x)^2 + (w_2-y)^2 \ge a$, and

\noindent
(iii) $\left| (w_1-x) + \xi(w_2-y)\right| \ge a$.

\noindent
Then (a) For each $v\in V$, every pair $(v,\ww)$ at frame distance $\le\eps$ is counted
(as an $\eps$-incidence of $v$) by the algorithm.
(b) For each $v\in V$, every pair $(v,\ww)$ that we count lies at frame distance $\le\alpha\eps$,
for some constant $\alpha>0$ depending on $a$.
\end{theorem}
%%%%%%%%%%%%%%%%%%%%%%%

\subsection{Running time analysis}
The cost of the algorithm is clearly proportional to
$
\sum_\tau \sum_{\tau^*} \left( m_{\tau^*} + n_{\tau^*} \right) ,
$
over all primal cells $\tau$ and the dual cells $\tau^*$ associated with each cell $\tau$. We have
\[
\sum_\tau \sum_{\tau^*} \left( m_{\tau^*} + n_{\tau^*} \right) =
O\left( \sum_\tau \left( m_\tau / \delta_2^3 + n_\tau \right) \right) =
O\left( m / \delta_2^3 + n / \delta_1^2 \right) .
\]
Optimizing the choice of $\delta_1$ and $\delta_2$, we choose
$\delta_1  = \left( \frac{\eps^3n}{m}\right)^{1/5}$ and
$\delta_2  = \left( \frac{\eps^2m}{n}\right)^{1/5}$.
These choices make sense as long as each of $\delta_1$, $\delta_2$ lies between $\eps/(2\gamma)$ and $1$.
That is, $\frac{\eps}{2\gamma}  \le \left( \frac{\eps^3n}{m}\right)^{1/5} \le 1$ and
$\frac{\eps}{2\gamma}  \le \left( \frac{\eps^2m}{n}\right)^{1/5} \le 1$, or
% Rearranging, both constraints translate to
$c'\eps^2 m \le n \le \frac{c''m}{\eps^3}$,
where $c'$ and $c''$ are absolute constants (that depend on $\gamma$).

If $n < c'\eps^2 m$, we use only the primal setup, taking $\delta_1 = \eps$ (for the primal subdivision). The cost is then
$
O\left( n/\eps^2 + m \right) = O\left( m \right) .
$
Similarly, if ${\displaystyle n >  \frac{c''m}{\eps^3} }$, we use only the dual setup,
taking $\delta_1 = 1$ and $\delta_2 = \eps/(2\gamma)$, and the cost is thus
$
O\left( n + m/\eps^3 \right) = O(n) .
$
Adding everything together, to cover all three subranges, the running time is then
$
O\left( \frac {m^{2/5}n^{3/5}}{\eps^{6/5}} + n + m \right) .
$
Substituting $m = O\left( 1/\eps^4 \right)$, we get a running time of
$
O\left( \frac {n^{3/5}}{\eps^{14/5}} + n + \frac{1}{\eps^4} \right) .
$
The first term dominates when $n=\Omega(\frac{1}{\eps^2})$ and $n=O(\frac{1}{\eps^7})$ .
In conclusion, we have the following result.
%%%%%%%%%%%%%%%%%%%%%%%%%%%%%
\begin{theorem} \label{main3}
Given $n$ data points that are seen (and identified) in a two-dimensional
image taken by a vertically positioned camera, and an error parameter $\eps>0$,
where the viewed points satisfy the assumptions made in Theorem~\ref{correct}, we can compute, in
${\displaystyle O\left( \frac {n^{3/5}}{\eps^{14/5}} + n + \frac{1}{\eps^4} \right)}$
time, a vertex $v$ of $G_\eps$ that maximizes the approximate count of $\eps$-incident 
correspondences, where ``approximate'' means that
every correspondence ${\bf w}$ whose surface ${\bf \sigma_w}$ is at frame distance at most $\eps$ from $v$ is counted and
every correspondence that we count lies at frame distance at most $\alpha\eps$ from $v$, for some fixed constant $\alpha$.
\end{theorem}

Restricting ourselves only to grid vertices does not really miss any solution. 
We only lose a bit in the quality of approximation, replacing $\eps$ by a slightly large 
constant multiple thereof, when we move from the best solution to a vertex of its grid cell.

%%%%%%%%%%%%%%%%%%%%%%%%%%%%%%%%%%%%%%%%%%%%%%%%%%%%%%%%%%%%%%%%%%%%%%%%%%%%%%%%%%%%%%%%%
\section{Geometric proximity via canonical surfaces}  \label{sec:fonseca} \label{fon:general}

In this section we present a general technique to preprocess a set of algebraic surfaces into
a data structure that can
 answer approximate incidences queries. In this technique we round the $n$ original  surfaces into
a set of canonical surfaces, whose size depends only on $\eps$, such that each original surface has
a canonical surface that is ``close'' to it.
Then we build an octree-based data structure for approximate incidences queries with respect to the canonical surfaces.
% This data structure is in fact an octree in which we store the canonical surfaces. 
However, to reduce the number of intersections  between the cells of the octree and the surfaces, we further reduce
the number of surfaces as  we go
 from one level of the octree to the next, by rounding them in a  coarser manner into
a smaller set of surfaces.

This technique has been introduced by  Fonseca and Mount~\cite{fonseca2010} for the case of hyperplanes.
% Since Fonseca and Mount did not provide full details of their technique, 
We describe as a warmup step, in Section \ref{fon:hyperplanes} of the appendix, 
our interpretation of their technique applied to hyperplanes. We then extend here the technique to general surfaces,
% in Section \ref{fon:general},
and apply it to the specific instance of 2-surfaces in 4-space that arise in the camera pose problem.

We have a set $S$ of $n$ $k$-dimensional surfaces in $\reals^d$
that cross the unit cube $[0,1]^d$, and a given error parameter $\eps$.
We assume that each surface $\sigma \in S$ is given in parametric form, where
the first $k$ coordinates are the parameters, so its equations are
\[
x_j = F_j^{(\sigma)}(x_1,\ldots,x_k) ,\qquad \text{for $j=k+1,\ldots,d$} .
\]
Moreover, we assume that each $\sigma\in S$ is defined in terms of \emph{$\ell$ essential
parameters} $\t=(t_1,\ldots,t_\ell)$, and $d-k$ additional
\emph{free additive parameters} $\ff=(f_{k+1},\ldots,f_d)$,
one free parameter for each dependent coordinate. Concretely, we assume that the equations defining the
surface $\sigma\in S$, parameterized by $\t$ and $\ff$
(we then denote $\sigma$ as $\sigma_{\t,\ff}$), are
\[
x_j = F_j(\xx;\t) + f_j = F_j(x_1,\ldots,x_k;t_1,\ldots,t_\ell) + f_j ,\qquad \text{for $j=k+1,\ldots,d$} .
\]

For each equation of the surface that does not have a free parameter in the original expression,
we introduce an {\em artificial} free parameter, and initialize its value to $0$.
(We need this separation into essential and free parameters for technical reasons
that will become clear later.) We assume that $\t$ (resp., $\ff$) varies over
% some compact $\ell$-dimensional domain $L$ (resp., $(d-k)$-dimensional domain $\Phi$)
%\haim{Dimension here can be lower if these are artificially added parameters},
% which, for simplicity of presentation, we take to be 
$[0,1]^\ell$ (resp., $[0,1]^{d-k}$).

\medskip
\noindent{\bf Remark.}
\noindent
The distinction between free and essential parameters seems to be artificial, but yet 
free parameters do arise in certain basic cases, such as the case of hyperplanes 
discussed in Section \ref{fon:hyperplanes} of the appendix. In the case of our 2-surfaces in 4-space,
the parameter $w_3$ is free, and we introduce a second artificial
free parameter into the equation for $\kappa$. The number of essential
parameters is $\ell=4$ (they are $w_1$,$w_2$,$\xi$, and $\eta$).

% \medskip
% \noindent
% 2) We assume that the parameterization of the surfaces is given as part of the input.
% Specifically, the input contains the definitions of the functions $F_j$ and the parameters $\ff$ for each surface.
% (Note that the parametrization of the surfaces is not unique, since
% we can, for example, add a constant to the function $F_j$ and subtract it from the corresponding free parameter $f_j$, for any surface and any $j$).

\medskip

We assume that the functions $F_j$ are all continuous and differentiable,
in all of their dependent variables $\xx$, $\t$ and $\ff$ (this is a trivial
assumption for $\ff$), and that they satisfy the following two conditions.

\medskip
\noindent{\bf (i) Bounded gradients.}
$
\left|\nabla_\xx F_j(\xx;\t) \right| \le c_1 ,\quad \left|\nabla_\t F_j(\xx;\t) \right| \le c_1 ,
$
for each $j=k+1,\ldots,d$, for any $\xx\in [0,1]^k$ and any
$\t\in [0,1]^\ell$, where $c_1$ is some absolute constant.
Here $\nabla_\xx$ (resp., $\nabla_\t$) means the gradient with respect to only
the variables $\xx$ (resp., $\t$).

\medskip
\noindent{\bf (ii) Lipschitz gradients.}
$
\left| \nabla_\xx F_j(\xx;\t) - \nabla_\xx F_j(\xx;\t') \right| \le c_2 |\t-\t'| ,
$
for each $j=k+1,\ldots,d$, for any $\xx\in [0,1]^k$ and any
$\t$, $\t'\in [0,1]^\ell$, where $c_2$ is some absolute constant.
This assumption is implied by the assumption that all the eigenvalues of the
mixed part of the Hessian matrix $\nabla_\t \nabla_\xx F_j(\xx;\t)$ have absolute value bounded by $c_2$.

%%%%%%%%%%%%%%%%%%%%%%%%%%%%%%%%%%%%%%%%%%%%%%%%%%%%%%%%%%%
\subsection{Canonizing the input surfaces}
\label{sec:canon}

We first replace each surface $\sigma_{\t,\ff} \in S$ by a canonical
``nearby'' surface $\sigma_{\ss,\g}$.
Let $\eps' = \frac{\eps}{c_2\log(1/\eps)}$ where $c_2$ is the constant from Condition (ii). 
We get $\ss$ from $\t$ (resp., $\g$ from $\ff$) by rounding 
each coordinate in the essential parametric domain $L$ (resp., 
in the parametric domain $\Phi$) to a multiple of $\eps'/(\ell+1)$.
Note that each of the artificial free parameters
(those that did not exist in the original equations) has the initial value $0$
for all surfaces, and remains $0$ in the rounded surfaces.
We get  $O\left( (1/\eps')^{\ell'} \right)$ {\em canonical} rounded surfaces, where
$\ell' \ge \ell$ is the number of \emph{original} parameters, that is, the
number of essential parameters plus the number of non-artificial free parameters;
in the worst case we have $\ell'=\ell+d-k$.

For a surface $\sigma_{\t,\ff}$ and its rounded version $\sigma_{\ss,\g}$ we have, for each $j$,
\begin{align*}
\left| \left( F_j(\xx;\t)+f_j\right) - \left(F_j(\xx;\ss)+g_j\right) \right|
 & \le |\nabla_\t F_j(\xx;\t')| \cdot  | \t - \ss | + |f_j-g_j| \\
& \le  c_1 | \t - \ss | + |f_j-g_j| \le (c_1+1) \eps' ,
\end{align*}
where $\t'$ is some intermediate value, which is irrelevant due to Condition (i).

We will use the $\ell_2$-norm of the difference vector
$\left( (F_j(\xx;\t)+f_j) - (F_j(\xx;\ss)+g_j) \right)_{j=k+1}^d$
as the measure of proximity between the
surfaces $\sigma_{\t,\ff}$ and $\sigma_{\ss,\g}$ at $\xx$, and denote it as
$\dist(\sigma_{\t,\ff},\sigma_{\ss,\g};\xx)$. The maximum
$\dist(\sigma_{\t,\ff},\sigma_{\ss,\g}) :=
\max_{\xx\in [0,1]^k} \dist(\sigma_{\t,\ff},\sigma_{\ss,\g};\xx)$ measures the
global proximity of the two surfaces. (Note that it is an upper bound on
the Hausdorff distance between the two surfaces.) We thus have
$\dist(\sigma_{\t,\ff},\sigma_{\ss,\g}) \le (c_1+1) \eps' $ when
$\sigma_{\ss,\g}$ is the canonical surface approximating $\sigma_{\t,\ff}$.

We define the {\em weight} of each
canonical surface to be the number of original surfaces that got rounded to it, 
and we refer to the set of all canonical surfaces by $S^c$.

%%%%%%%%%%%%%%%%%%%%%%%%%%%%%%%%%%%%%%%%%%%%%%%%%%%%%%%%%%%
\subsection{Approximately counting $\eps$-incidences}

We describe an algorithm for approximating the $\eps$-incidences counts of the surfaces in $S$
and the vertices of a grid $G$ of side length $4\eps$.

We construct an octree decomposition of $\tau_0 := [0,1]^d$, all the way to subcubes of side length $4\eps$
such that each vertex of $G$ is the center of a leaf-cube. We propagate the surfaces of $S^c$ down this octree,
further rounding each of them within each subcube that it crosses.

\looseness=-1
The root of the octree corresponds to $\tau_0$, and we set $S_{\tau_0} =S^c$.
At level $j\ge 1 $ of the recursion, we have  subcubes $\tau$ of $\tau_0$
of side length $\delta=1/2^j$.
For each such $\tau$, we set $\tilde{S}_\tau$ to be the subset of the surfaces
in $S_{p(\tau)}$ (that have been produced at the parent cube $p(\tau)$ of $\tau$)
that intersect $\tau$. We now show how to further round the
surfaces of $\tilde{S}_\tau$, so as to get a coarser set $S_\tau$ of surfaces
that we associate with $\tau$, and that we process recursively within $\tau$.

At any node $\tau$ at level $j$
of our rounding process, each surface $\sigma$ of
$S_\tau$ is of the form
 $x_j = H_j(\xx;\t) + f_j$, for $j=k+1,\ldots,d$ where $\xx = (x_1,\ldots,x_k)$, and ${\bf t} = (t_1,\ldots,t_\ell)$.

\smallskip
\noindent
(a) For each $j=k+1,\ldots,d$ the function $H_j$ is a translation of $F_j$. That is
$H_j(\xx;\t) = F_j(\xx;\t)+c$ for some constant $c$. Thus the gradients of $H_j$ also satisfy Conditions (i) and (ii).

\smallskip
\noindent (b) $\t$ is some vector of $\ell$ essential parameters, and each coordinate of $\t$ is an integer multiple of
$\frac{\eps'}{(\ell+1)\delta}$, where
 $\delta = 1/2^j$.

\smallskip
\noindent (c) $\ff=(f_{k+1},\ldots,f_d)$ is a vector of free parameters, each is a multiple of $\eps'/(\ell+1)$.

Note that the surfaces in $S_{\tau_0} =S^c$, namely the set of initial canonical surfaces
constructed in Section \ref{sec:canon},
are of this form (for $j=0$ and $H_j = F_j$).
We get $S_\tau$ from $\tilde{S}_\tau\subseteq S_{p(\tau)}$ by the following steps.
The first step just changes the presentation of $\tau$ and
$\tilde{S}_\tau$, and the following steps do the actual rounding to obtain
$S_\tau$.

\begin{enumerate}
\item
Let  $(\xi_1,\ldots,\xi_k,\xi_{k+1},\ldots,\xi_d)$ be the point in $\tau$ of smallest coordinates and set $\xxi = (\xi_1,\ldots,\xi_k)$.
We rewrite the equations of each surface of $\tilde{S}_\tau$ as follows:
 $x_j = G_j(\xx;\t) + f'_j$, for $j=k+1,\ldots,d$, where
 $G_j(\xx;\t) = H_j(\xx ;\t) - H_j(\xxi ;\t) +\xi_j $, and
$f'_j =  f_j + H_j(\xxi ;\t) - \xi_j$, for $j=k+1,\ldots,d$. Note that in this reformulation we have not changed the essential parameters, but we
did change the free parameters from $f_j$ to $f'_j$, where $f'_j$ depends on $f_j$, $\t$, $\xxi$, and
$\xi_j$. Note also that $G_j(\xxi;\t)=\xi_j$ for $j=k+1,\ldots,d$.
\item
We replace the essential parameters $\t$ of a
surface $\sigma_{\t,\ff}$ by $\ss$, which we obtain by rounding each
coordinate of $\t$ to the nearest integer multiple of $\frac{\eps'}{(\ell+1)\delta}$.
So the  rounded surface has the equations
$x_j = G_j(\xx;\ss) + f'_j$, for $j=k+1,\ldots,d$.
Note that we also have that $G_j(\xxi;\ss)=\xi_j$, for $j=k+1,\ldots,d$.

\item
For each surface,  we round each  free parameter $f'_j$, $j=k+1,\ldots,d$,
to an integral multiple of $\frac{\eps'}{\ell+1}$, and denote the rounded vector by $\g$.
Our final equations for each rounded surface that we put in $S_\tau$ are
$x_j = G_j(\xx;\ss)  + g_j$ for $j=k+1,\ldots,d$.

\end{enumerate}

By construction, when $\t_1$ and $\ff'_1$
and  $\t_2$ and $\ff'_2$ get rounded to the same vectors $\ss$ and $\gg$
then the corresponding two surfaces in $\tilde{S}_\tau$ get rounded to the same surface
in $S_\tau$.
The weight of each surface in $S_\tau$ is the sum of
 the weights of the surfaces in $S_{p(\tau)}$
that got rounded to it, which,
 by induction, is the number of  original surfaces that are recursively rounded to it.
In the next step of the recursion
the $H_j$'s of the parametrization of the surfaces in $S_\tau$ are the functions $G_j$ defined above.

The total weight of the surface in $S_\tau$ for a leaf cell $\tau$ is the approximate $\eps$-incidences count that
we associate with the center of $\tau$.

\subsection{Error analysis}

We now bound the error incurred by our discretization.
We start with the following lemma, whose proof is given in Appendix~\label{proofs}.

%%%%%%%%%%%%%%%%%%%%%%%%%%%%%%%%%%
\begin{lemma} \label{mixed:hessian}
Let $\tau$ be a cell of the octtree and let $x_j = G_j(\xx;\t) + f'_j$, for $j=k+1,\ldots,d$ be a surface
obtained in Step 1 of the rounding process described above.
For any $\xx=(x_1,\ldots,x_k)\in [0,\delta]^k$, for any $\t,\ss \in [0,1]^\ell$,
and for each $j=k+1,\ldots,d$, we have
\begin{equation} \label{hdiff}
\left| G_j(\xx;\s) - G_j(\xx;\t) \right| \le c_2|\xx-\xxi|\cdot|\t-\ss| ,
\end{equation}
where $c_2$ is the constant of Condition (ii), and $\xxi = (\xi_1,\ldots,\xi_k)$ consists of the first $k$
coordinates of
the point in $\tau$ of smallest coordinates.
\end{lemma}

%%%%%%%%%%%%%%%%%%%%%%%%%%%%%%%%
\begin{lemma} \label{thm:error}
For any $\xx=(x_1,\ldots,x_k)\in [0,\delta]^k$, for any $\t$, $\ss\in [0,1]^\ell$,
and for each $j=k+1,\ldots,d$, we have
\begin{equation} \label{hdiff1}
\left| G_j(\xx;\s) +g_j - (G_j(\xx;\t) + f'_j) \right| \le c_2\eps' \le \frac{\eps}{\log(1/\eps)},
\end{equation}
where $c_2$ is the constant of Condition (ii).
\end{lemma}
%%%%%%%%%%%%%%%%%%%%%%%%%%%%%%%%
%\noindent{\bf Proof.}
\begin{proof}
Using the triangle inequality and Lemma \ref{mixed:hessian}, we get that
\begin{align*}
\left| G_j(\xx;\s) +g_j - (G_j(\xx;\t) + f'_j) \right| &
\le \left| G_j(\xx;\s) - G_j(\xx;\t) \right| + \left| g_j - f'_j \right| \le c_2 |\xx-\xxi| |\t-\ss| + \frac{\eps'}{\ell+1} \ .
\end{align*}
Since $|\xx-\xxi| \le \delta$, $|\t-\ss| \le \frac{\ell\eps'}{(\ell+1)\delta}$, and $| g_j - f'_j | \le \frac{\eps'}{\ell+1}$, the lemma follows.
\end{proof}

We now bound the number of surfaces in $S_\tau$.
Since $\ss\in [0,1]^\ell$ and each of its coordinates is a multiple of $\frac{\eps'}{(\ell+1)\delta}$, we have at most
$(\frac{\delta}{\eps'})^\ell$ different values for $\ss$.
To bound the number of possible values of $\g$, we prove the following lemma (see the appendix for the proof).

%%%%%%%%%%%%%%%%%%%%%%%%%%%%%%%%%
\begin{lemma} \label{const:snap}
Let
 $x_j = G_j(\xx;\t) + f'_j$, for $j=k+1,\ldots,d$, be a surface $\sigma_{\t,\ff'}$ in $\tilde{S}_\tau$.
For each $j=k+1,\ldots,d$, we have
$\left|  f'_j \right| \le (c_1+1)\delta$,
where $c_1$ is the constant of Condition (i).
\end{lemma}
%%%%%%%%%%%%%%%%%%%%%%%%%%%%%%%%%

Lemma \ref{const:snap} implies that each $g_j$, $j=k+1,\ldots,d$, has
only $O(\frac{\delta}{\eps'})$ possible values, for a total of at most $O((\frac{\delta}{\eps'})^{d-k})$
possible values for $\g$.
% \micha{Do the artificial free parameters help here?}
Combining the number of possible values for $\s$ and $\g$, we get that
the number of newly discretized surfaces in $S_\tau$ is
\begin{equation} \label{num:new}
O\left( \left( \frac{\delta}{\eps'} \right)^\ell \cdot \left(\frac{\delta}{\eps'} \right)^{d-k} \right)
= O\left( \left( \frac{\delta}{\eps'} \right)^{\ell+d-k} \right) .
\end{equation}

%\subsubsection{Overall complexity bound}
%The first phase of the procedure generates
%$
%O\left( \left( \frac{1}{\eps'} \right)^{\ell'} \right)
%$
%canonical surfaces, where $\ell'$ is the number of
%essential parameters plus the number of non-artificial free parameters;
%as observed, $\ell \le \ell'\le \ell + d -k$.

It follows that each level of the recursive octree decomposition generates
\[
O\left( \left(\frac{1}{\delta}\right)^d \cdot
\left( \frac{\delta}{\eps'} \right)^{\ell+d-k} \right)
= O\left( \frac{\delta^{\ell-k}}{(\eps')^{\ell+d-k}} \right)
\]
re-discretized surfaces, where the first factor in the left-hand side expression
is the number of cubes generated at this recursive level, and the second factor
is the one in (\ref{num:new}).

Summing over the recursive levels $j=0,\ldots,\log\frac{1}{\eps}$, where the cube size $\delta$
is $1/2^j$ at level $j$, we get a total size of
$
O\left( \frac{1}{(\eps')^{\ell+d-k}} \sum_{j=0}^{\log\frac{1}{\eps}} \frac{1}{2^{j(\ell-k)}} \right)$.
We get different estimates for the sum according to the sign of $\ell-k$.
If $\ell>k$ the sum is $O(1)$. If $\ell=k$ the sum is $O\left(\log\frac{1}{\eps}\right)$.
If $\ell<k$ the sum is
$
O\left( 2^{j_{\rm max}(k-\ell)} \right) = O\left( \frac{1}{(\eps')^{k-\ell}} \right) .
$
Accordingly, the overall size of the structure, taking also into account the cost of the first phase, is
\begin{equation} \label{size:bound}
\begin{cases}
O\left(\frac{1}{(\eps')^{\ell+d-k}} \right) & \text{for $\ell > k$} \\
O\left(\frac{1}{(\eps')^{d}}\log\frac{1}{\eps} \right) & \text{for $\ell = k$} \\
O\left(\frac{1}{(\eps')^{d}} \right) & \text{for $\ell < k$} .
\end{cases}
\end{equation}

The following theorem summarizes the result of this section. Its proof follows in a straightforward way from the preceding discussion
from Lemma \ref{thm:error},
analogously to the proof of Lemma~\ref{const:snap} in the appendix.

\begin{theorem} \label{thm:fon-surf}
Let $S$ be a set of $n$ surfaces
in $\reals^d$ that cross the unit cube $[0,1]^d$,
given parametrically as $x_j = F_j(\xx;\t) + f_j$ for $j=k+1,\ldots,d$, where the functions $F_j$ satisfy conditions (i) and (ii),
and $\t = (t_1,\ldots, t_\ell)$.
Let $G$ be the  $(4\eps)$-grid within $[0,1]^d$.
The algorithm described above reports for each vertex $v$ of $G$ an approximate $\eps$-incidences count that includes all
surfaces at distance at most $\eps$ from $v$ and may include some surfaces at distance at most $(2\sqrt{d}+1)\eps$ from $v$.
The running time of this algorithm is proportional to the total number of rounded surfaces that it generates, which is given by
Equation (\ref{size:bound}), plus an additive $O(n)$ term for the initial canonization of the surfaces.
\end{theorem}

We can modify our data structure so that it can answer approximate or exact $\eps$-incidence queries as we
describe in Section \ref{fon:hyperplanes} of the appendix for the case of hyperplanes.

%%%%%%%%%%%%%%%%%%%%%%%%%%%%%%%%%%%%%%%%%%%%%%%%%%%%%%%

%\pagestyle{fancy}
%\fancyhead[RE,RO]{}
%\fancyhead[LE,LO]{}

\section{Experimental Results} \label{sec:experiments}

The goal of the experimental results is to show the practical relation between the naive, the primal-dual and the general canonical surfaces algorithms. It is not our intention to obtain the fastest possible code, but to obtain a platform for fair comparison between the techniques. We have performed a preliminary experimental comparison using synthetic as well as real-world data. We focus on values of $n,\eps$ that are practical in real applications. Typically, we have $100K$-$200K$ $3$D points bounded by a rectangle of size 100-150 meters and the uncertainty is around 3m (so the relative error is $\eps~=0.03$). The three methods that we evaluate are:
\begin{itemize}
\item The naive method, with asymptotic run-time  $O(\frac{n}{\eps^2})$.
\item The primal-dual method (cf. Section \ref{sec:primal_dual}), with asymptotic run-time $O(n+\frac{n^{3/5}}{\eps^{14/5}}+\frac{1}{\eps^4})$.
\item The canonical surfaces method (cf. Section \ref{fon:general}), with asymptotic run-time $\tilde{O}(n+\frac{1}{\eps^6})$ (ignoring poly logarithmic factors).
\end{itemize}

In all experiments we normalize the data, so that the camera position ($x,y,z$) and the $3$D points lie in the unit box $[0,1]^3$, and the forth parameter ($\kappa$) representing the camera orientation lies in $[-1,1]$.

\subsection{Random synthetic data}Starting from a fixed known camera pose, we generate a set of $n$ uniformly sampled $3$D points which are projected onto the camera image plane using Eq. (\ref{cweqs}). To model outliers in the association process we use random projections for 90\% of the $3D$ points, resulting in an inlier ratio of $10\%$. We add Gaussian noise of zero mean and $\sigma=0.02$ to the coordinates of each $3D$ point. This provides us with $2$D-$3$D correspondences that are used for estimating the camera pose.
We apply the three algorithms above and measure the run-times, where each algorithm is tested for its ability to reach approximately the (known) solution. We remark that the actual implementation may be slowed down by the (constant) cost of some of its primitive operations, but it can also gain efficiency from certain practical heuristic improvements. For example, in contrast to the worst case analysis, we could stop the recursion in the algorithm of Section~\ref{sec:fonseca}, at any step of the octree expansion, whenever the maximum incidence count obtained so far is larger than the number of surfaces crossing a cell of the octree. The same applies for the primal-dual technique in the dual stage. On the other hand, finding whether or not a surface crosses a box in pose space, takes at least the time to test for intersections of the surface with 32 edges of the box, and this constant affects greatly the run-time. The $O(1/\eps^6)$ bound in the canonical surfaces algorithm is huge and has no effect in practice for this problem. For this reason, the overall number of surfaces that we have to consider in the recursion can be very large. The canonical surfaces algorithm in our setting does not change much with $\eps$ because we are far from the second term effect. We show in Figure~\ref{fig:asympth_vs_real}, a comparison of the three algorithms.

\begin{figure}[!h]
\centering
\captionsetup{justification=centering,margin=2cm}
\begin{minipage}{.32\linewidth}
    \includegraphics[width=\linewidth]{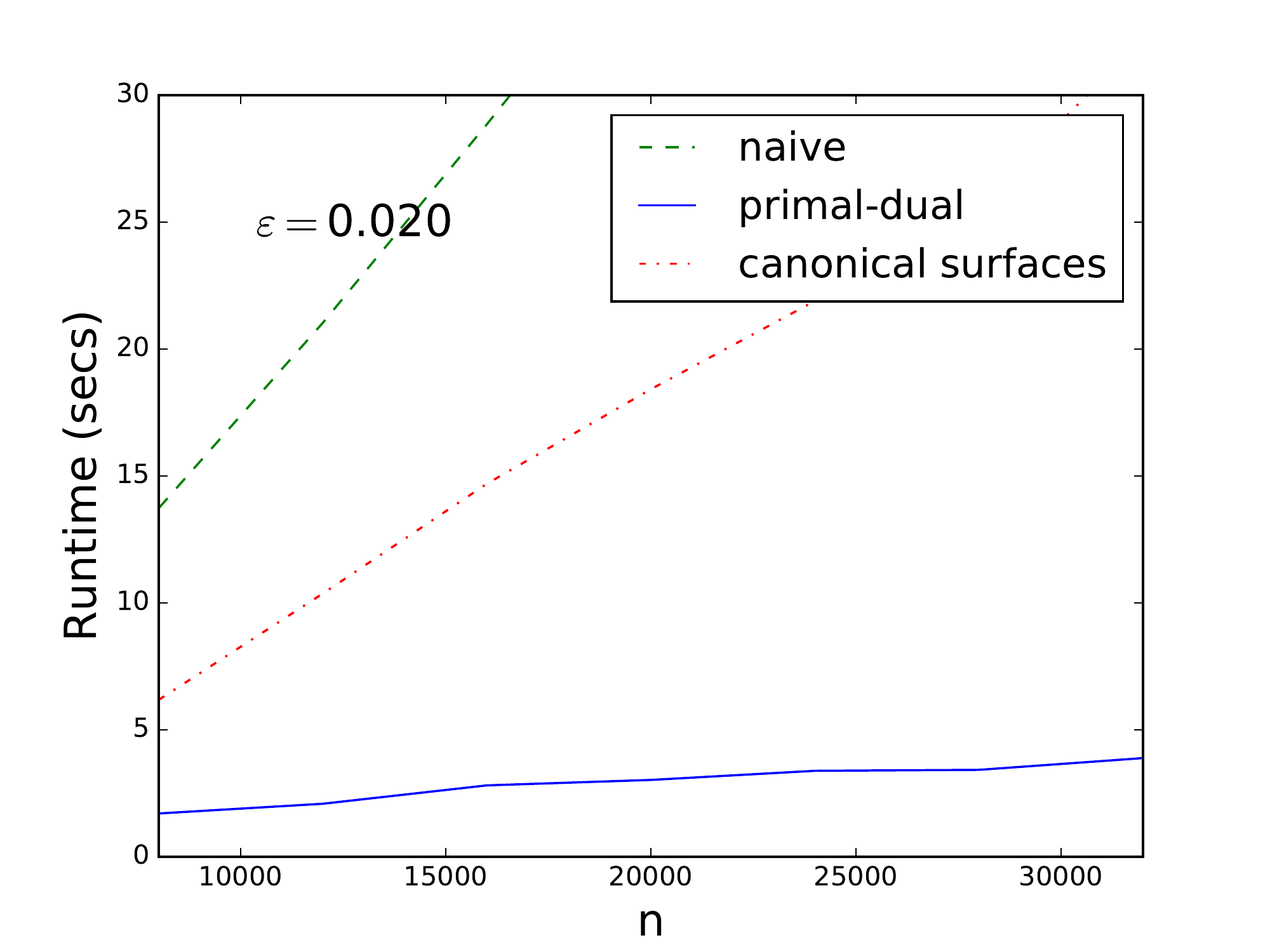}
\end{minipage}
\begin{minipage}{.32\linewidth}
    \includegraphics[width=\linewidth]{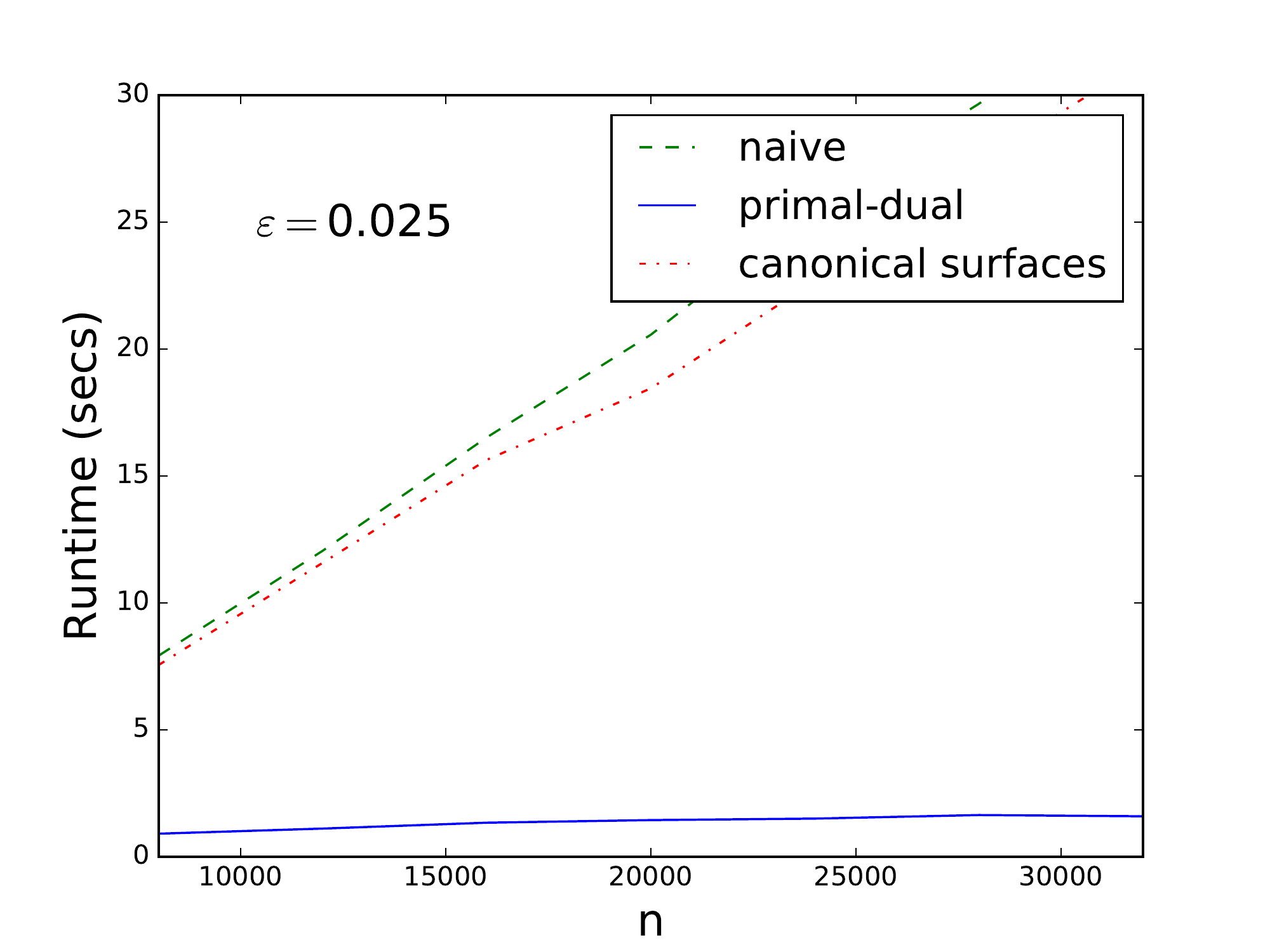}
\end{minipage}
\begin{minipage}{.32\linewidth}
    \includegraphics[width=\linewidth]{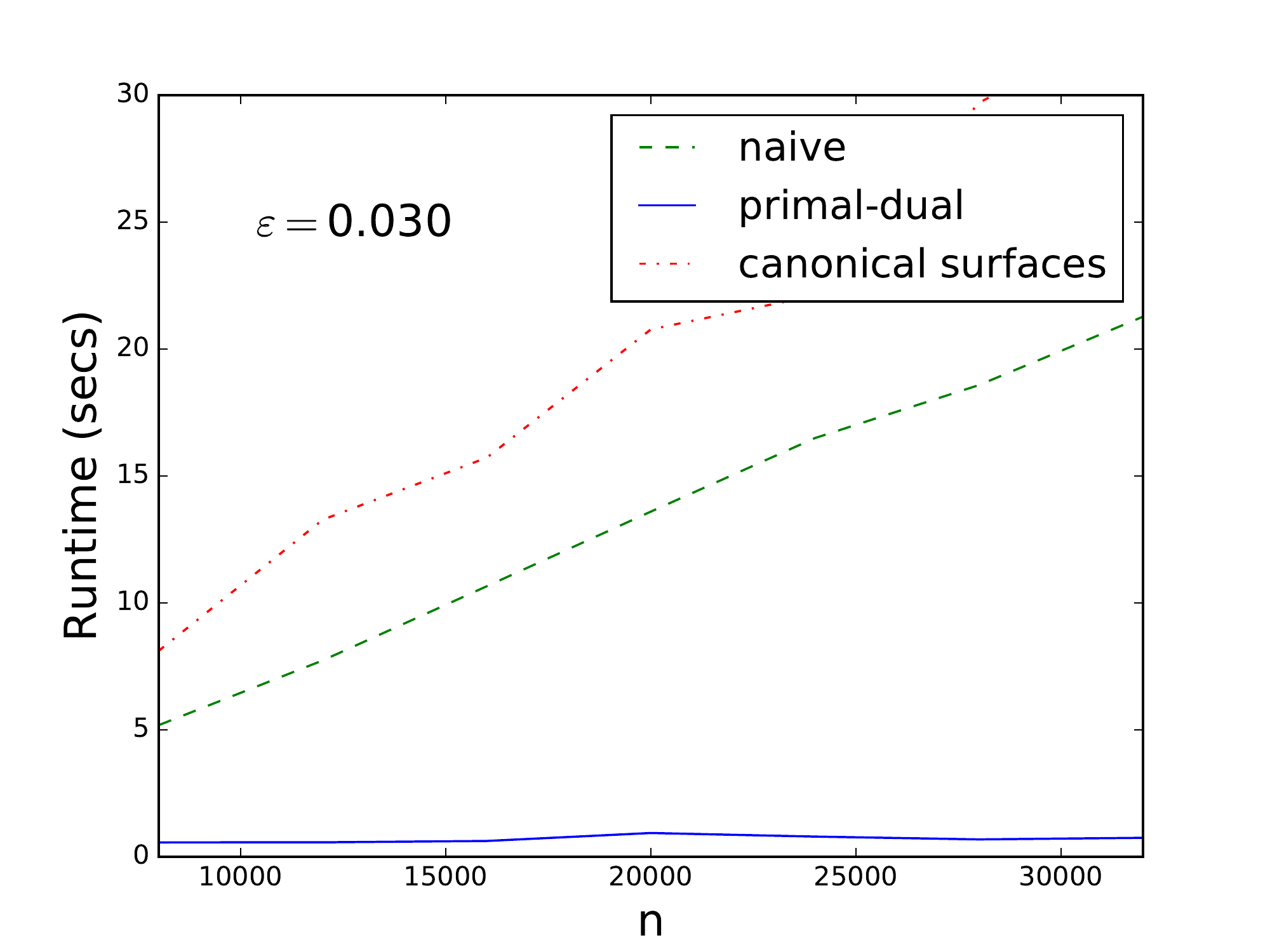}
\end{minipage}
\center
\caption{The run-time of the three methods for various values of $\eps$.}
\label{fig:asympth_vs_real}
\end{figure}

The computed camera poses corresponding to Figure \ref{fig:asympth_vs_real},
obtained by the three algorithm for various problem sizes, are displayed
in Table \ref{tab:poses}, compared to the known pose. The goal here is not to obtain the most accurate algorithm but to show that they are comparable in accuracy in this setting so the runtime comparison is fair.

\begin{table}[h]
\begin{center}
 \begin{tabular}{||c||c c c c||}
 \hline
 n & x(N/PD/C) & y(N/PD/C) & z(N/PD/C) & $\kappa$(N/PD/C) \\ [0.5ex]
 \hline\hline
 8000 & 0.31/0.31/0.28 & 0.22/0.2/0.18 & 0.1/0.12/0.09 & 0.55/0.66/0.59 \\
 \hline
 12000 & 0.31/0.33/0.28 & 0.22/0.17/0.19 & 0.1/0.1/0.1 & 0.55/0.65/0.6 \\
 \hline
 24000 & 0.31/0.3/0.28 & 0.22/0.2/0.18 & 0.1/0.1/0.09 & 0.55/0.61/0.59 \\
 \hline\hline
 32000 & 0.31/0.27/0.28 & 0.22/0.2/0.19 & 0.1/0.08/0.09 & 0.55/0.57/0.59 \\
 \hline\hline
 True pose & 0.3 & 0.2 & 0.1 & 0.6 \\
 \hline
\end{tabular}
\end{center}
\captionsetup{justification=centering,margin=2cm}
\caption{Poses computed by the three algorithms for $\eps=0.03$ and various problem sizes (N:naive, PD:primal-dual, C:canonical surfaces).}
\label{tab:poses}
\end{table}

\subsection{Real-world data}We evaluated the performance of the algorithms also on real-world datasets
for which the true camera pose is known. The input is a set of correspondences, each
represented by a $5$-tuple $(w_1,w_2,w_3,\xi,\eta)$, where $(w_1,w_2,w_3)$ are the
$3D$ coordinates of a salient feature in the scene and $(\xi,\eta)$ is its corresponding projection
in the camera frame. We computed the camera pose from these matches using both
primal-dual and naive algorithms and compared the poses to the true one. An example of the data we have used is
shown in Figure \ref{fig:realdata_input}.

\begin{figure}[h]
\centering
\captionsetup{justification=centering,margin=2cm}
\begin{minipage}{.32\linewidth}
    \centering
    \includegraphics[width=\linewidth]{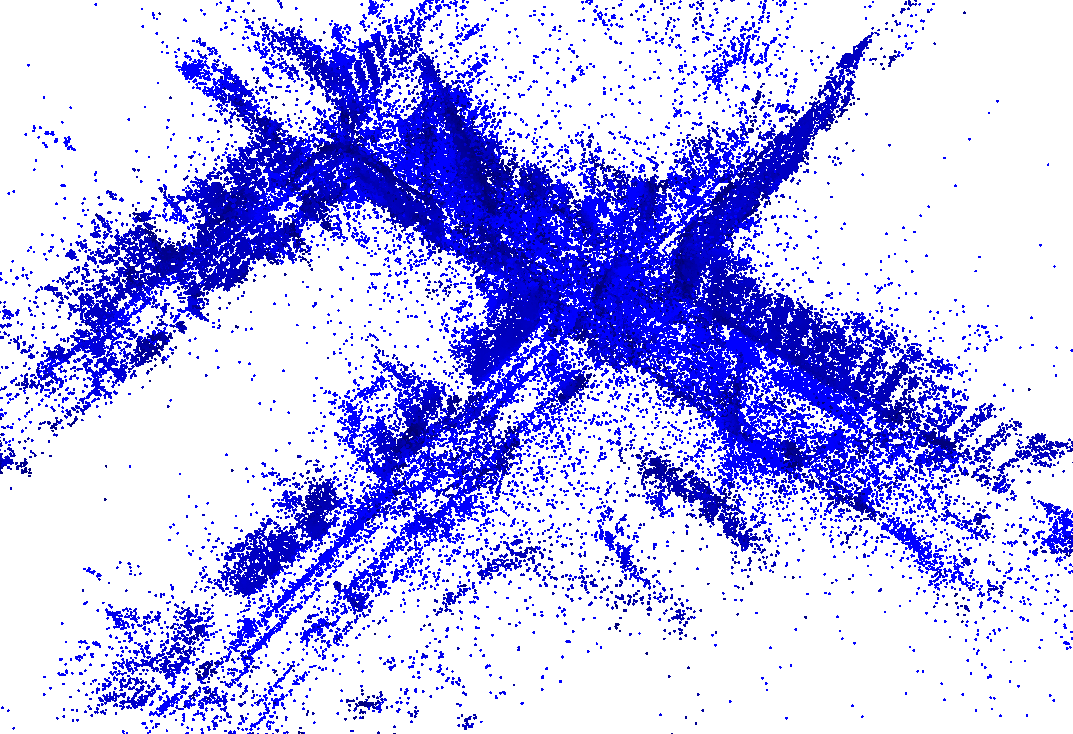}
    \subcaption{\textsl{140000 $3$D landmarks in correspondence to image features in ~\ref{fig:realdata_frame}.}}
    \label{fig:realdata_3d}
\end{minipage}
\begin{minipage}{.32\linewidth}
    \centering
    \includegraphics[height=\linewidth]{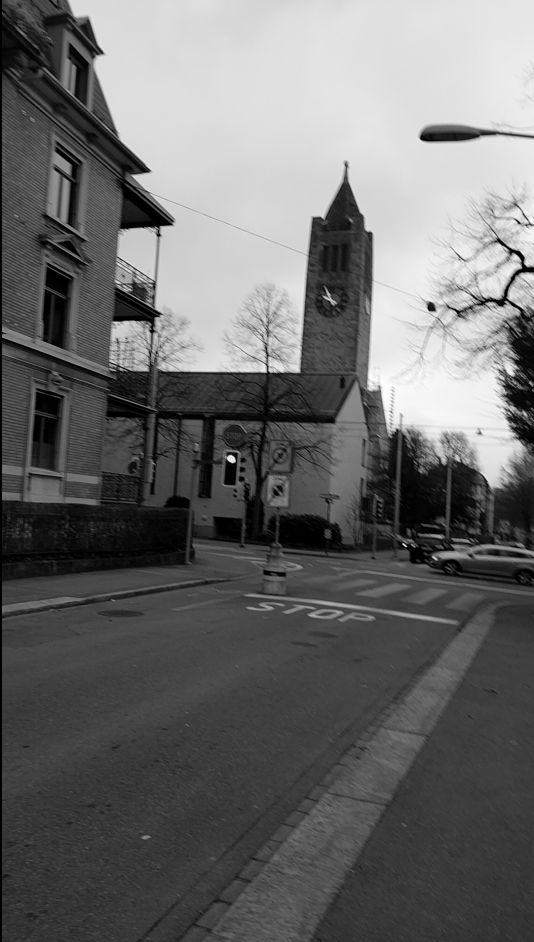}
    \subcaption{\textsl{Query image}.}
    \label{fig:realdata_query}
\end{minipage}
\begin{minipage}{.32\linewidth}
    \centering
    \includegraphics[width=\linewidth]{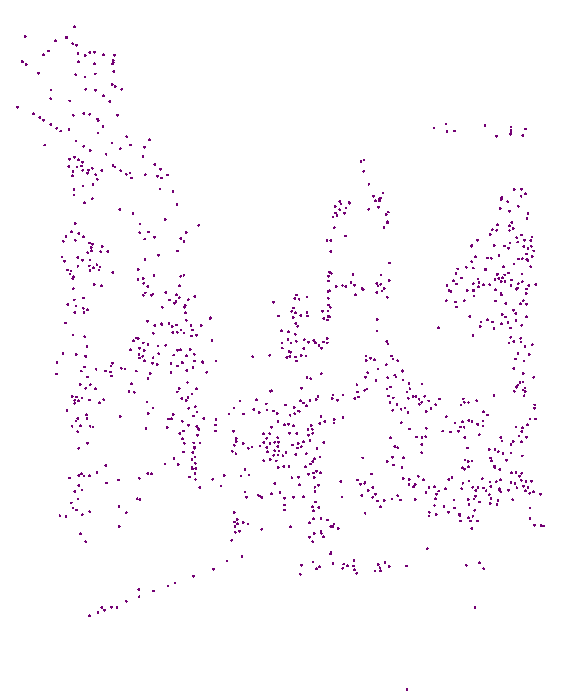}
    \subcaption{\textsl{Corresponding $(\xi,\eta)$ points, each corresponds to many landmarks.}}
    \label{fig:realdata_frame}
\end{minipage}
\begin{minipage}{.32\linewidth}
    \centering
    \includegraphics[width=\linewidth]{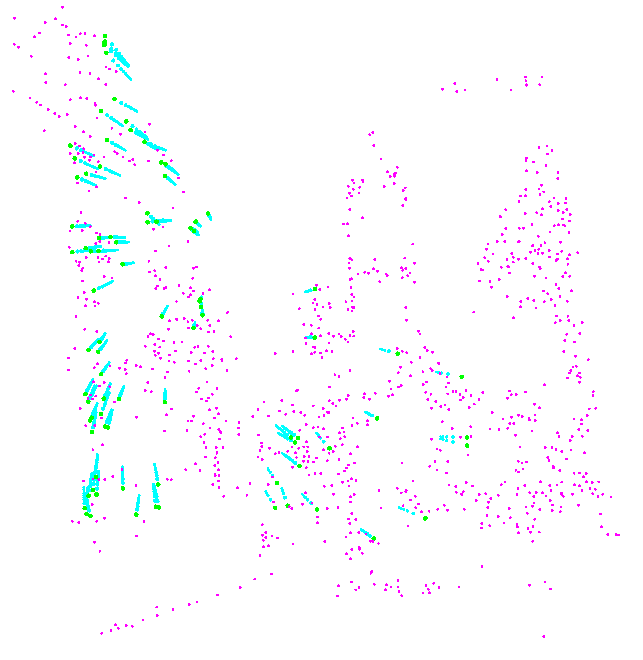}
    \subcaption{\textsl{Inliers ($(\xi,\eta)$ with $\eps$-close projection of landmarks) found by the algorithm (green) along with the rays from matched landmarks}}
    \label{fig:realdata_inliers}
\end{minipage}
\begin{minipage}{.32\linewidth}
    \centering
    \includegraphics[width=\linewidth]{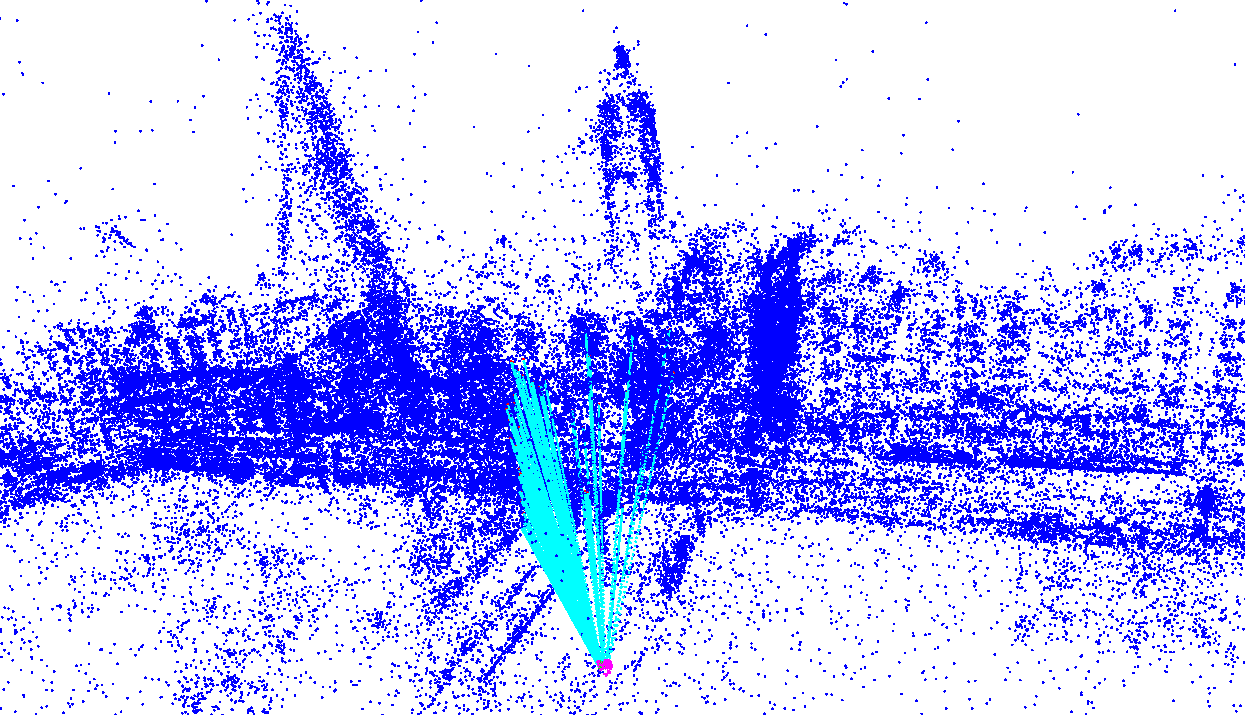}
    \subcaption{\textsl{The pose found by the algorithm (within 3m and 10 degrees from the known true pose) with rays to matched landmarks.}}
    \label{fig:realdata_pose}
\end{minipage}
\caption{Real-world data input and pose}
\label{fig:realdata_input}
\end{figure}

We evaluated the runtime for different problem sizes and checked the correctness of the
camera pose approximation when the size increased. To get different input sizes,
we added random correspondences to a base set of actual correspondences. The number of random correspondences
determines the input size but also the fraction of good correspondences (percentage of
inliers) which goes down with increased input size (the number of inliers in real world cases is typically 10\%). We show the same plots as before in Figure \ref{fig:realdata_plot} and Table
\ref{tab:realdata_poses}.

\begin{figure}[h]
\centering
\includegraphics[width=0.32\linewidth]{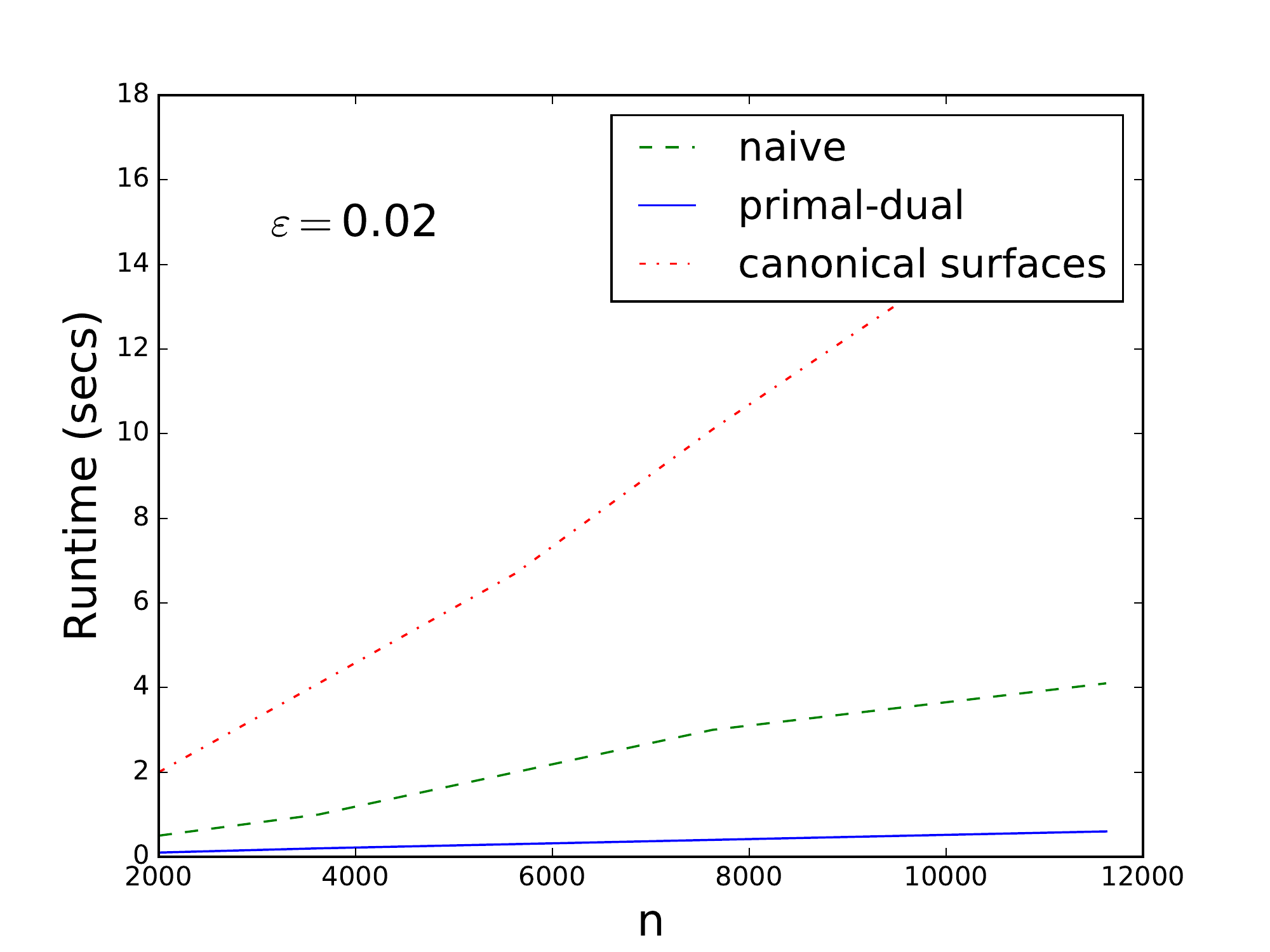}
\caption{Runtime for real-world data with increased $n$ and decreased number of inliers.}
\label{fig:realdata_plot}
\end{figure}

\begin{table}[h!]
\begin{center}
 \begin{tabular}{||c||c c c c||}
 \hline
 n & x & y & z & $\kappa$ \\ [0.5ex]
 \hline\hline
 2000 & 0.26 & 0.62 & 0.06 & 0.72 \\
 \hline
 3626 & 0.36 & 0.56 & 0.12 & 0.52 \\
 \hline
 5626 & 0.38 & 0.60 & 0.12 & 0.62 \\
 \hline
 7626 & 0.40 & 0.64 & 0.08 & 0.57 \\
 \hline
 11626 & 0.43 & 0.68 & 0.13 & 0.63 \\
 \hline\hline
 True pose & 0.37 & 0.59 & 0.06 & - \\
 \hline
\end{tabular}
\end{center}
\captionsetup{justification=centering,margin=2cm}
\caption{Poses computed by the primal-dual algorithms for real-world data (we do not know the actual orientation here).}
\label{tab:realdata_poses}
\end{table}

\section{Future work} \label{sec:conclusion}
%In this work we have proposed to transfer the classical problem of camera pose estimation to a problem of estimating approximate $\eps$-incidences between regular points on a grid (representing pose hypotheses) and surfaces (corresponding to $2$D-$3$D correspondences).
%In this setting the best pose estimation corresponds to the point where most surfaces are $\eps$-close, which closely relates to the classical scoring of pose hypotheses by counting the number of matches that exhibit a reprojection error smaller than a certain threshold.
%In the well established RANSAC-based filtering scheme for pose estimation the run-time is growing with $O(n^2)$.
%In contrast, our approach enables us to compute a camera pose in $O(n)$ (plus some constant), that is our approach is also efficient in the case when the number of correspondences n is large.
%
%We have presented three different algorithms to solve the $\eps$incidence problem.
%The first one simply counts the number of intersections between each surface and grid cell.
%In our second algorithm we propose to use a geometric primal-dual formulation which is asymptotically faster than the simple algorithm for $n \ge 1/\eps^2$.
%For our third algorithm we leverage a data structure that approximates the input surface by a fixed set of canonical surfaces and then recurses on the problem resolution.
%We have formulated and analyzed this technique in its full generality, which is why we believe it should be of general interest, also beyond the application of camera posing.

We note that similar approaches can be applied for computing the relative pose~\cite{nister2004} between two cameras (that look at the same scene), except that the pose estimation then uses $2$D-$2$D matches between the two images (rather than 2D-3D image-to-model correspondences).
Determining the relative motion between images is a prerequisite for stereo depth estimation~\cite{scharstein2002}, in multi-view geometry~\cite{hartley2003}, and for the initialization of the view graph~\cite{sweeney2015vg,zach2010} in SfM, and is therefore an equally important task in computer vision.
In addition, in future work we want to also consider the case of a generalized or distributed camera setup and likewise transform the camera posing problem to an $\eps$-incidence problem.

\pagestyle{empty}
\bibliography{main_journal}
\bibliographystyle{abbrv}

\appendix

\section{Omitted proofs}
\label{proofs}
In all experiments we normalize the data, so that the camera position ($x,y,z$) and the $3$D points lie in the unit box $[0,1]^3$, and the forth parameter ($\kappa$) representing the camera orientation lies in $[-1,1]$.
Let
$\xi_v := F(v;w)$,
We apply the three algorithms above and measure the run-times, where each algorithm is tested for its ability to reach approximately the (known) solution. We remark that the actual implementation may be slowed down by the (constant) cost of some of its primitive operations, but it can also gain efficiency from certain practical heuristic improvements. For example, in contrast to the worst case analysis, we could stop the recursion in the algorithm of Section~\ref{sec:fonseca}, at any step of the octree expansion, whenever the maximum incidence count obtained so far is larger than the number of surfaces crossing a cell of the octree. The same applies for the primal-dual technique in the dual stage. On the other hand, finding whether or not a surface crosses a box in pose space, takes at least the time to test for intersections of the surface with 32 edges of the box, and this constant affects greatly the run-time. The $O(1/\eps^6)$ bound in the canonical surfaces algorithm is huge and has no effect in practice for this problem. For this reason, the overall number of surfaces that we have to consider in the recursion can be very large. The canonical surfaces algorithm in our setting does not change much with $\eps$ because we are far from the second term effect. We show in Figure~\ref{fig:asympth_vs_real}, a comparison of the three algorithms.
Since $v'\in \sigma_\ww$ we have that
$\xi  = F(v';w)$,
$\eta  = G(v';w)$.
We want to show that
$
\fd(v,\ww) = \max \left\{ |\xi_v-\xi|,\; |\eta_v-\eta| \right\} \le \beta \eps
$
for some constant $\beta$ that depends on $a$.

Regarding $F$ and $G$ as functions of $v$, we compute their gradients as follows.
\begin{align*}
F_x & = \frac{ \kappa\left[ (w_1-x)+\kappa(w_2-y)\right] + \left[ (w_2-y)-\kappa(w_1-x)\right] }
{\left( (w_1-x)+\kappa(w_2-y)\right)^2}
= \frac{(1+\kappa^2)(w_2-y)}
{\left( (w_1-x)+\kappa(w_2-y)\right)^2} \\
F_y & = \frac{ - \left[ (w_1-x)+\kappa(w_2-y)\right] + \kappa\left[ (w_2-y)-\kappa(w_1-x)\right] }
{\left( (w_1-x)+\kappa(w_2-y)\right)^2}
= - \frac{(1+\kappa^2)(w_1-x)}
{\left( (w_1-x)+\kappa(w_2-y)\right)^2} \\
F_z & = 0 \\
F_\kappa & = \frac{ -(w_1-x)\left[ (w_1-x)+\kappa(w_2-y)\right] -(w_2-y) \left[ (w_2-y)-\kappa(w_1-x)\right] }
{\left( (w_1-x)+\kappa(w_2-y)\right)^2} \\
& = - \frac{(w_1-x)^2+(w_2-y)^2 } {\left( (w_1-x)+\kappa(w_2-y)\right)^2} ,
\end{align*}
and
\begin{align*}
G_x & = \frac{(w_3-z)(w_1-x)}
{\left((w_1-x)^2+(w_2-y)^2\right)^{3/2}} \\
G_y & = \frac{(w_3-z)(w_2-y)}
{\left((w_1-x)^2+(w_2-y)^2\right)^{3/2}} \\
G_z & = -\frac{1}
{\left((w_1-x)^2+(w_2-y)^2\right)^{1/2}} \\
G_\kappa & = 0 .
\end{align*}
Conditions (i) and (ii) in the lemma, plus the facts that we restrict both
$(w_1,w_2,w_3)$ and $(x,y,z)$ to lie in the bounded domain $[0,1]^3$, and
that $|\kappa|$ is also at most $1$, are then easily seen to imply that the $v$-gradients
$|\nabla F|,\; |\nabla G|$ are at most $\beta$, for some constant $\beta$ that
depends on $a$, and so
$|\xi_v-\xi|  \le \beta |v-v'| \le \beta \eps$ and
$|\eta_v-\eta|  \le \beta |v-v'| \le \beta \eps$,
and the lemma follows.

\medskip
\begin{proof}[Proof of Lemma \ref{reverse-primal}:]Let
\begin{align}
\xi_v & = \frac{(w_2-y) - \kappa (w_1-x)} {(w_1-x) + \kappa (w_2-y)} \label{zpkp} \\
\eta_v & = \frac{w_3-z}{\sqrt{(w_1-x)^2+(w_2-y)^2}} \nonumber .
\end{align}
Since $\fd(v,\ww) \le\eps$ we have that $|\xi_v-\xi|,\;|\eta_v-\eta| \le\eps$.

Since the Equations (\ref{cparw}) are the inverse system of
those of \eqref{cweqs}, we can rewrite (\ref{zpkp}) as
\begin{align*}
z & = w_3 - \eta_v \sqrt{(w_1-x)^2+(w_2-y)^2} \\
\kappa & = \frac{(w_2-y) - \xi_v (w_1-x)} {(w_1-x) + \xi_v (w_2-y)} .
\end{align*}
Hence
\begin{align*}
z-z' & = (\eta - \eta_v) \sqrt{(w_1-x)^2+(w_2-y)^2} \\
\kappa - \kappa' & = \frac{(w_2-y) - \xi_v (w_1-x)} {(w_1-x) + \xi_v (w_2-y)} -
\frac{(w_2-y) - \xi (w_1-x)} {(w_1-x) + \xi (w_2-y)} .
\end{align*}
It follows right away that $|z-z'| \le \sqrt{2}\eps$ (recall that
all the points lie in the unit cube). For the other difference, writing
\[
H(t) =
\frac{(w_2-y) - t (w_1-x)} {(w_1-x) + t (w_2-y)}
\]
(with the other parameters being fixed), we get
\[
|\kappa - \kappa'| \le \max_{t\in [\xi_v,\xi]} |H'(t)| |\xi_v-\xi| .
\]
As is easily verified, we have
\[
H'(t) = - \frac{(w_1-x)^2+(w_2-y)^2}{\left[ (w_1-x) + t (w_2-y) \right]^2} .
\]
Since $\left| (w_1-x) + \xi(w_2-y)\right| \ge a>0$, and $|\xi_v-\xi|\le\eps$,
the denominator of $H'(t)$ is bounded away from zero (assuming that
$\eps$ is sufficiently small), and $|H'(t)| \le c$ for
$t\in[\xi,\xi_v]$, where $c$ is some fixed positive constant. This
implies that $|\kappa - \kappa'| \le c\eps$, and the lemma follows.
\end{proof}

\medskip
%\noindent{\bf Proof of Theorem \ref{correct}.}
\begin{proof}[Proof of Theorem \ref{correct}. Part (a):]
\noindent
%{\bf Part (a):}
Let $(v,\ww)$ be a pair at frame distance $\le\eps$. By Lemma \ref{reverse-primal}
and the definition of $G_{\delta_1}$, there exists a cell $\tau \in G_{\delta_1}$
such that $v\in \tau$ and $\ww \in S_\tau$.

By definition, the surface $\sigma^*_{v;\tau}$ contains the point
\[
(w_1,w_2,w_3,\xi_v-F(c_\tau;w),\eta_v-G(c_\tau;w)) ,
\]
where $\xi_v$ and $\eta_v$ are given by (\ref{zpkp}).
Since $\fd(v,\ww)\le \eps$, the points $(w_1,w_2,w_3,\xi_v-F(c_\tau;w),\eta_v-G(c_\tau;w))$ and
$(w_1,w_2,w_3,\xi_\tau,\eta_\tau)$ lie at $L_\infty$-distance at most $\eps$, therefore
$\sigma^*_v\in V^*_{\tau^*}$ where $\tau^* \in G_{\delta_2}$ is the cell that contains $\ww_\tau$.

Together, these two properties imply that $(v,\ww)$ is counted by the algorithm.
Moreover, since we kept both primal and dual boxes pairwise disjoint, each such pair
is counted exactly once.

\noindent
\medskip
{\bf Part (b):}
Let $(v,\ww)$ be an $\eps$-incident pair that we encounter, where $v$ and $\ww$
are encoded as above. That is, $\sigma_\ww$ crosses the primal cell $\tau$ of
$G_{\delta_1}$ that contains $v$, or a neighboring cell in the $(z,\kappa)$-directions,
and $\sigma^*_{v;\tau}$ crosses the dual cell $\tau^*$ that contains $\ww_\tau$, or a
neighboring cell in the $(\xi_\tau,\eta_\tau)$-directions. This means that $\tau$
(or a neighboring cell) contains a point $c=(x',y',z',\kappa') \in\sigma_\ww$,
and $\tau^*$ (or a neighboring cell) contains a point
$\ww'_\tau=(w_1',w_2',w_3',\xi'_\tau,\eta'_\tau) \in\sigma^*_{v;\tau}$.
The former containment means that
\[
\xi  = \frac{(w_2-y') - \kappa' (w_1-x')} {(w_1-x') + \kappa' (w_2-y')}, \quad\quad
\eta  = \frac{w_3-z'}{\sqrt{(w_1-x')^2+(w_2-y')^2}} ,
\]
and that
\[
|x-x'|, |y-y'| \le \delta_1, \quad\quad
|z-z'| \le 2\sqrt{2}\delta_1, \quad\text{and}\quad
|\kappa-\kappa'| \le 2c\delta_1 .
\]
To interpret the latter containment, we write, using the definition of $\xi'_\tau$, $\eta'_\tau$,
and the fact that $\ww'_\tau \in\sigma^*_{v;\tau}$,
\begin{align*}
\xi'_\tau & = F(v;w') - F(c_\tau;w')
= \frac{(w'_2-y) - \kappa (w'_1-x)} {(w'_1-x) + \kappa (w'_2-y)}
- \frac{(w'_2-y_\tau) - \kappa_\tau (w'_1-x_\tau)} {(w'_1-x_\tau) + \kappa_\tau (w'_2-y_\tau)} \\
\eta'_\tau & = G(v;w') - G(c_\tau;w'_1)
= \frac{w'_3-z}{\sqrt{(w'_1-x)^2+(w'_2-y)^2}}
- \frac{w'_3-z_\tau}{\sqrt{(w'_1-x_\tau)^2+(w'_2-y_\tau)^2}} ,
\end{align*}
where $w'=(w'_1,w'_2,w'_3)$, and where
$c_\tau = (x_\tau,y_\tau,z_\tau,\kappa_\tau)$ is the centerpoint of $\tau$, and
\begin{align*}
\max\left\{ |w_1-w_1'|,\; |w_2-w_2'|,\; |w_3-w_3'| \right\} & = 2\delta_2 \\
\max\left\{ |\xi_\tau-\xi'_\tau|,\; |\eta_\tau-\eta'_\tau| \right\} & = 2\eps ,
\end{align*}
where
\[
\xi_\tau = \xi - F(c_\tau;w) \quad\quad\text{and}\quad\quad
\eta_\tau = \eta - G(c_\tau;w) .
\]
By definition (of $F$, $G$, and the frame distance), we have
\[
\fd(v,\ww) = \max\left\{
\left| \xi - F(v;w) \right|,\;
\left| \eta - G(v;w) \right|
\right\} ,
\]
which we can bound by writing
\begin{align*}
\left| \xi - F(v;w) \right| & \le
\left| \xi - F(v;w') + F(c_\tau;w') - F(c_\tau;w) \right| \\
& + \left| F(v;w) - F(v;w') + F(c_\tau;w') - F(c_\tau;w) \right| \\
& = \left| \xi_\tau - \xi'_\tau \right|
+ \left| F(v;w) - F(v;w') + F(c_\tau;w') - F(c_\tau;w) \right| , \\
\left| \eta - G(v;w) \right| & \le
\left| \eta - G(v;w') + G(c_\tau;w') - G(c_\tau;w) \right| \\
& + \left| G(v;w) - G(v;w') + G(c_\tau;w') - G(c_\tau;w) \right| \\
& = \left| \eta_\tau - \eta'_\tau \right|
+ \left| G(v;w) - G(v;w') + G(c_\tau;w') - G(c_\tau;w) \right| .
\end{align*}
We are given that
\[
\left| \xi_\tau - \xi'_\tau \right| ,\;
\left| \eta_\tau - \eta'_\tau \right| \le 2\eps ,
\]
so it remains to bound the other term in each of the two right-hand sides.
Consider for example the expression
\begin{equation} \label{fders}
F(v;w) - F(v;w') + F(c_\tau;w') - F(c_\tau;w) .
\end{equation}
Write $c_\tau = v+t$ and $w'=w+s$, for suitable vectors $t,s\in\reals^3$.
We expand the expression up to second order, by writing
\begin{align*}
F(v;w') & = F(v;w+s) = F(v,w) + s\cdot \nabla_w F(v;w) + \frac12 s^T H_w(v;w) s \\
F(c_\tau;w) & = F(v+t;w) = F(v,w) + t\cdot \nabla_v F(v;w) + \frac12 t^T H_v(v;w) t \\
F(c_\tau;w') & = F(v+t;w+s) = F(v,w) + s\cdot \nabla_w F(v;w) + t\cdot \nabla_v F(v;w) \\
& + \frac12 s^T H_w(v;w) s + \frac12 t^T H_v(v;w) t + t^T H_{v;w}(v;w) s ,
\end{align*}
where $\nabla_w$ (resp., $\nabla_v$) denotes the gradient with respect to the variables
$w$ (resp., $v$), and where $H_w$ (resp., $H_v$, $H_{v;w}$) denotes the Hessian submatrix
of second derivatives in which both derivatives are with respect to $w$ (resp., both are
with respect to $v$, one derivative is with respect to $v$ and the other is with respect to $w$).

Substituting in (\ref{fders}), we get that, up to second order,
\begin{align*}
| F(v;w) - F(v;w') & + F(c_\tau;w') - F(c_\tau;w) | \\
& = |t^T H_{v;w}(v;w) s| \le \|H_{v;w}(v;w)\|_\infty |t||s| ,
\end{align*}
where $\|H_{v;w}(v;w)\|_\infty$ is the maximum of the absolute values of all the ``mixed'' second derivatives.
(Note that the mixed part of the Hessian of the Hessian arises also in the analysis of the algorithm in
Section \ref{fon:general}.)
% \micha{Connect also in Fonseca's method ?}
Arguing as in the preceding analysis and using the assumptions in the theorem, one can show
that all these derivatives are bounded by some absolute constants, concluding that
\[
| F(v;w) - F(v;w') + F(c_\tau;w') - F(c_\tau;w) | = O(\delta_1\delta_2) = O(\eps) ,
\]
which implies that
\[
\left| \xi - F(v;w) \right| = O(\eps) .
\]
Applying an analogous analysis to $G$, we also have
\[
\left| \eta - G(v;w) \right| = O(\eps) .
\]
Together, these bounds complete the proof of part (b) of the theorem.
\end{proof}

%%%%%%%%%%%%%%%%%%%%%%%%%%%%%%%%%%
\medskip
%\noindent{\bf Proof of Lemma \ref{mixed:hessian}.}
\begin{proof}[Proof of Lemma \ref{mixed:hessian}:]
Fix $j$, consider the function
\[
K_j(\xx) := G_j(\xx;\ss)-G_j(\xx;\t) ,
\]
and recall our assumption that $G_j(\xi;\t) = G_j(\xi;\ss)=\xi_j$.
Then we can also write the left-hand side of (\ref{hdiff}) as $K_j(\xx)-K_j(\xxi)$.
By the intermediate value theorem, it can be written as
\[
K_j(\xx)-K_j(\xxi) = \langle \nabla K_j(\xx') ,\xx-\xxi \rangle ,
\]
for some intermediate value $\xx'$ (that depends on $\ss$ and $\t$).
By definition, we have,
\begin{equation} \label{yyp}
\nabla K_j(\xx') =
\nabla_\xx G_j(\xx';\ss) - \nabla_\xx G_j(\xx';\t) ,
\end{equation}
whose norm is bounded by $c_2|s-t|$ by Condition (ii).
Using the Cauchy-Schwarz inequality, we can thus conclude that
\begin{align*}
\left| G_j(\xx;\ss) - G_j(\xx;\t) \right| & =
\left| K_j(\xx)-K_j(\xxi) \right| \\
& \le |\xx-\xxi|\cdot \left| \nabla_\xx G_j(\xx';\ss) - \nabla_\xx G_j(\xx';\t) \right|  \\
& \le  c_2 |\xx-\xxi| |\t-\ss|  \ ,
\end{align*}
as asserted.
\end{proof}

\medskip
%\noindent{\bf Proof of Lemma \ref{const:snap}.}
\begin{proof}[Proof of Lemma \ref{const:snap}:]
Each surface $\sigma_{\t,\ff'}$ in $\tilde{S}_\tau$ meets $\tau$. That is, there exists a point
$(x_1,\ldots,x_d)$ in $\tau = [0,\delta]^d$ that lies on $\sigma_{\t,\ff'}$, so we have
$\xi_j \le G_j(\xx;\t) + f'_j \le \xi_j + \delta$ for each $j=k+1,\ldots,d$, where $\xx=(x_1,\ldots,x_k)$.
Hence, for some intermediate value $\xx'$, we have
\begin{align*}
\left| f'_j \right| & =
\left| G_j(\xxi;\t) - f'_j - G_j(\xx;\t) + G_j(\xx;\t)- G_j(\xxi;\t) \right|  \\
& \le \left| \xi_j - (f'_j + G_j(\xx;\t))  \right| + \left| G_j(\xx;\t)- G_j(\xxi;\t) \right| \\
& \le \delta + \left| \nabla_\xx G_j(\xx';\t) \right| \cdot |\xx-\xxi| \\
& \le \delta + c_1\delta = (c_1+1)\delta ,
\end{align*}

where the first inequality follows by the triangle inequality, the second follows since
$\left| \xi_j - (f'_j + G_j(\xx;\t))  \right|\le \delta$, the third by the intermediate value theorem and the Cauchy-Schwarz inequality, and the fourth
by Condition (i).
\end{proof}

%%%%%%%%%%%%%%%%%%%%%%%%%%%%%%%
\section{A simple algorithm} \label{sec:simple}

We  present a simple naive solution which does not require any of the sophisticated machinery developed in this work. It actually turns out to be the most efficient solution when $n$ is small.

We construct a grid $G$ over $Q=[0,1]^3\times[-1,1]$, of cells $\tau$, each of
dimensions $\eps \times \eps \times 2\sqrt{2}\eps \times 2c\eps$,
where $c$ is the constant of Lemma \ref{reverse-primal}. (We use this non-square grid $G$ since we want to find $\eps$-approximate incidences in
terms of frame distance.)
For each  cell $\tau$ of $G$ we compute the number of surfaces $\sigma_\ww$ that intersect $\tau$.

Consider now a shifted version $G'$ of $G$ in which the vertices of $G'$ are the centers of
the cells of $G$.
To report how many surfaces are within frame distance $\eps$ from a vertex $q\in G'$,
we return the count of the cell of $G$ whose center is $q$.
By Lemma \ref{reverse-primal}
and Lemma \ref{xetau}, this includes all surfaces at frame distance $\eps$ from $q$,
but may also count surfaces at frame distance at most $\sqrt{10+4c^2}\beta\eps$ from $q$,
where $\beta$ is the constant in Lemma \ref{xetau}. (The distance from $q$ to the farthest  corner of
 its cell is  $\sqrt{1^2 + 1^2 + (2c)^2 + (2\sqrt{2})^2}\eps=\sqrt{10+4c^2}\beta\eps$.)

It takes $O(\frac{n}{\eps^2})$ time to construct this data structure. Indeed, cell
 boundaries reside on $O(\frac{1}{\eps})$ hyperplanes, so we
compute the intersection curve of
 each surface with each of these hyperplanes, in a total of $O(\frac{n}{\eps})$  time. Then, for each such curve we find the cell boundaries
 that it intersects  within its three-dimensional hyperplane
in $O(\frac{1}{\eps})$ time.
We summarize this result in the following theorem.

\begin{theorem}
The algorithm described above
approximates the the number of surfaces that are at distance $\eps$ to each vertex $q\in G'$
where $G'$ is  an $\eps \times \eps \times 2\sqrt{2}\eps \times 2c\eps$ grid in
 $O(\frac{n}{\eps^2})$. (The approximation is in the sense defined above.)
\end{theorem}
\begin{proof}
Correctness follow from Lemmas \ref{reverse-primal} and \ref{xetau}, the running time follows
since there are only $O(\frac{n}{\eps^2})$ cells of $G$ that at least one surface intersects.
\end{proof}

In fact we can find for each vertex $q$ of  $G'$, the {\em exact}
 number of $\eps$-incident surfaces (i.e.\ surfaces at distance
at most $\eps$ from $q$). For this we keep with each cell $\tau$ of $G$, the list of the surfaces that intersect $\tau$.
Then for each vertex $q \in G'$ we traverse the surfaces stored in its cell and check which of them is within frame distance
$\eps $ from $q$. The asymptotic running time is still $O(\frac{n}{\eps^2})$.

If we want to get an  incidences counts of vertices of   a finer grid that $G$, we use a union of several shifted grids as above.
This also allows to construct a data structure that can return an  $\eps$-incidences count of any query point.

For the camera pose problem we use the
 vertex of $G_\eps$ of largest $\eps$-incidences count as the position of the camera.

%%%%%%%%%%%%%%%%%%%%%%%%%%%%%%%

\section{Geometric proximity via canonical surfaces: The case of hyperplanes} \label{fon:hyperplanes}
We have a set $H$ of $n$ hyperplanes in $\reals^d$
that cross the unit cube $\tau_0=[0,1]^d$, and a given error parameter $\eps$.
Each hyperplane $h \in H$ is given by an equation of the form
$x_d = \sum_{i=1}^{d-1} a_ix_i + b$.
We assume, for simplicity, that $|a_i|\le 1$ for each $h\in H$ and for each $i=1,\ldots,d-1$.
Moreover, since $h$ crosses $\tau_0$, we have $|b|\le d$, as is easily checked. (This can always
be enforced by rewriting the equation turning the $x_i$ with the coefficient $a_i$
of largest absolute value into the independent coordinate.)

For our rounding scheme we define $\eps' = \eps/\log(1/\eps)$.
We discretize each hyperplane $h\in H$ as follows. Let the equation of $h$ be
$x_d = \sum_{i=1}^{d-1} a_ix_i + b$.
We replace each $a_i$ by the integer multiple of $\eps'/d$ that is nearest to it,
and do the same for $b$. Denoting these `snapped' values as $a'_i$ and $b'$,
respectively, we replace $h$ by the hyperplane $h'$, given by
$x_d = \sum_{i=1}^{d-1} a'_ix_i + b'$.
For any $\xx=(x_1,\ldots,x_{d-1})\in[0,1]^{d-1}$, the $x_d$-vertical
distance between $h$ and $h'$ at $\xx$ is
\[
\left| \left( \sum_{i=1}^{d-1} a_ix_i + b\right) -
\left( \sum_{i=1}^{d-1} a'_ix_i + b'\right) \right|
\le \sum_{i=1}^{d-1} |a_i-a'_i|x_i + |b-b'|
\le \sum_{i=1}^{d-1} |a_i-a'_i| + |b-b'| \le \eps' .
\]
We define the {\em weight} of each
canonical hyperplane to be the number of original hyperplanes that got rounded to it, and  we refer to the set of
all canonical hyperplanes by $H^c$.

We describe a recursive procedure that approximates
the number of $\eps$-incident hyperplanes of $H$ to each  vertex of a $(4\eps)$-grid $G$ that tiles up $[0,1]^d$.
Specifically, for
 each vertex $v$ of $G$ we report a count that includes all  hyperplanes
 in $H$ that are at Euclidean distance at most $\eps$ from $v$ but it may also count hyperplanes of $H$ that are at distance up to
$(2\sqrt{d}+1)\eps$ from $v$.

Our procedure  constructs an octree decomposition of $\tau_0$, all the way to subcubes
of  side length $4\eps$. (We assume that $4\eps$ is a negative power of $2$ to avoid rounding issues.)
We shift  the grid $G$ such that its vertices are centers of these leaf-subcubes.
At level $j$ of the recursive construction, we have  subcubes $\tau$
of side length $\delta=1/2^j$.
For each such $\tau$ we construct
a set $H_\tau$ of more coarsely  rounded hyperplanes.
The weight of each
 hyperplane $h$ in  $H_\tau$
is the sum of the
weights of the hyperplanes in the parent cube $p(\tau)$ of $\tau$ that got rounded to $h$, which, by induction, is the number of
 original hyperplanes that are rounded to it (by repeated rounding along the path in the recursion tree leading to $\tau$).

At the root, where $j=0$, we set $H_\tau = H^c$ (where each $h\in H_\tau$ has the initial weight of the number
of original hyperplanes rounded to it, as described above). At any other cell $\tau$ we
obtain $H_\tau$ by applying a rounding step to the set $\tilde{H}_\tau$ of
the hyperplanes of $H_{p(\tau)}$ that intersect $\tau$.

The coarser
discretization of the hyperplanes of $\tilde{H}_\tau$
that produces the set $H_\tau$ proceeds as follows.
Let $(\xi_1,\ldots,\xi_d)$ denote the coordinates of the corner
of $\tau$ with smallest coordinates,
so $\tau = \prod_{i=1}^d [\xi_i,\xi_i+\delta]$.

Let $h$ be a hyperplane of $H_\tau$, and rewrite its equation as
\[
x_d - \xi_d = \sum_{i=1}^{d-1} a_i(x_i - \xi_i) + b \ .
\]
This rewriting only changes the value of $b$ but does not
affect the $a_i$'s. Since $h$ crosses $\tau$, we have
$|b|\le d\delta$ (and $|a_i|\le 1$ for each $i$).
We now re-discretize each coefficient $a_i$ (resp., $b$) to the integer
multiple of $\frac{\eps'}{d\delta}$ (resp., $\frac{\eps'}{d}$) that is nearest to it.
Denoting these snapped values as $a'_i$ and $b'$, respectively,
we replace $h$ by the hyperplane $h'$ given by
\[
x_d - \xi_d = \sum_{i=1}^{d-1} a'_i(x_i - \xi_i) + b' .
\]
This re-discretization of the coefficients $a_i$ is a
coarsening of the discretization of the hyperplanes in $\tilde{H}_\tau$.
The set $H_\tau$ contains all the new, more coarsely rounded hyperplanes that we obtain from the hyperplanes in
$\tilde{H}_\tau$ in this manner.
Note that several hyperplanes in $\tilde{H}_\tau$ may be rounded to the same hyperplane in $H_\tau$.
We set the weight of each hyperplane in $H_\tau$ to be the sum of the weights of the hyperplanes in $\tilde{H}_\tau$
that got rounded to it.
(Note that although every hyperplane of $\tilde{H}_\tau$ crosses $\tau$, such an $h$ may get rounded to a hyperplane
that misses $\tau$, in which case it is not represented by any hyperplane in $H_\tau$.)

For any $\xx=(x_1,\ldots,x_{d-1})\in \prod_{i=1}^{d-1} [\xi_i,\xi_i+\delta]$,
the $x_d$-vertical distance between $h$ and $h'$ at $\xx$ is
\begin{align*}
& \Bigl| \Bigl( \xi_d + \sum_{i=1}^{d-1} a_i(x_i-\xi_i) +b \Bigr) -
\Bigl( \xi_d + \sum_{i=1}^{d-1} a'_i(x_i-\xi_i) +b' \Bigr) \Bigr| \\
& \le \sum_{i=1}^{d-1} |a_i-a'_i|(x_i-\xi_i) + |b-b'| \le
\sum_{i=1}^{d-1} |a_i-a'_i| \delta + |b-b'| \le \eps' .
\end{align*}
Since the original value of $a_i$ is in $[-1,1]$ and we round it to an integer  multiple of
$\frac{\eps'}{d\delta}$, the hyperplanes in $H_\tau$ have $O(\frac{\delta}{\eps'})$ possible values for each coefficient $a_i$. Furthermore, these hyperplanes also have  $O(\frac{\delta}{\eps'})$ possible values for $b$,
 because $ |b|\le \delta d$ for every hyperplane in
$\tilde{H}_\tau$ (since it intersects $\tau$).
It follows that $|H_\tau| = O\left(\left(\frac{\delta}{\eps'}\right)^d\right)$, and
 the total size of
all sets $H_\tau$, over all cells  $\tau$ at the same level of the octree, is
$O\left(\left(\frac{1}{\eps'}\right)^d\right)$.

Finally, at every leaf
$\tau$ of the octree we report the sum of the weights of the hyperplanes in $H_\tau$ as the
approximate $\eps$-incidences
  count of the vertex of $G$ at the center of $\tau$.

\begin{theorem} \label{thm:fon-hyper}
Let $H$ be a set of $n$ hyperplanes
in $\reals^d$ that cross the unit cube $[0,1]^d$,
and let $G$ be the  $(4\eps)$-grid within $[0,1]^d$.
The algorithm described above reports for each vertex $v$ of $G$ an approximate $\eps$-incidences count that includes all
hyperplanes at Euclidean distance at most $\eps$ from $v$ and may include some hyperplanes at distance at most $(2\sqrt{d}+1)\eps$ from $v$.
The running time of this algorithm is $O\left(n+\frac{(\log(1/\eps))^{d+1}}{\eps^d}\right)$.
\end{theorem}
\begin{proof}
Let $v\in G$, and consider a hyperplane $h\in H$ at distance at most $\eps$ from $v$.
The hyperplane $h$ is rounded to a hyperplane $h'\in H^c$ which is at  distance at most $\eps'=\frac{\eps}{\log(1/\eps)}$ from $h$,\footnote{The distance between two hyperplanes is defined to be the maximum vertical distance between them.} and thereby at
distance at most $\eps +\eps'$ from $v$. The hyperplane $h'$ is further rounded to other hyperplanes while
propagating down the octree. The distance from $h'$
from the hyperplane that it is rounded to in $H^c$ is at most ${\eps}{\log(1/\eps)}$, and, in general the distance of $h'$ from
 any hyperplane
$h_j$, that it is rounded to at any level $j$, is  at most
$(j+1)\eps'=\frac{(j+1)\eps}{\log(1/\eps)}$.
(Note that $h'$ is rounded to different hyperplanes in different cells of level $j$.)
Therefore the distance of $h_j$ from $v$ is
at most $\eps+\frac{(j+1)\eps}{\log(1/\eps)}\le 2\eps$ (since $j+1\le \log(1/\eps)$). It follows that $h'$ is rounded
to some hyperplane that crosses the cell that contains $v$, at each level of the octree.
In particular $h'$ is (repeatedly) rounded to some hyperplane at the leaf containing $v$ and  is included
in the weight of some hyperplane at that leaf.

Consider now a hyperplane $h\in H$ that is rounded to some hyperplane $h^\ell$ at the leaf $\tau$ containing $v$.
The hyperplane $h^\ell$ is at distance at most $\eps$ from   $h$.
Therefore $h$ is at distance at most $\eps$ from the boundary of the leaf-cell containing $v$.
The distance of $v$ to the boundary of the leaf-cell containing it is at most $2\sqrt{d}\eps$,
so the distance of $h$ from $v$ is at most $(2\sqrt{d}+1)\eps$.

The running time follows from the fact that the total size of the sets $H_{\tau}$ for all
cells $\tau$ at a particular level of the quadtree is $O\left(\frac{(\log(1/\eps))^{d}}{\eps^d}\right)$ and there are
$\log(1/(4\eps))$ levels.
\end{proof}

Note that if we consider an arbitrary point $p$ then each
hyperplane at distance at most $\eps$ from $p$ is included in the approximate count of at least one of the
 vertices of the grid $G$ surrounding $p$.
In this rather weak sense, the  largest approximate incidences count of a  vertex  of $G$
can be considered as an approximation to the number of $\eps$-close hyperplanes to
 the point $p\in \reals^d$ with the largest number of $\eps$-close hyperplanes.

Our octree data structure can give an
approximate  $\eps$-incidences count for any query point $q$ (albeit with somewhat worse constants).
For this we construct a constant number of octree structures over $5^d$ shifted (by intergral multiple of $\eps$) grids of
a somewhat larger side-length, say $5\eps$.
The grids are shifted such that each cell $c$ of a finer  grid of side length $\eps$ is centered in
a  larger grid cell of one of our grids, say $G_c$. We use $G_c$ to answer queries $q$ that lie in $c$,
by returning the sum of the weights of the hyperplanes in $h_\tau$ where $\tau$ is the leaf of $G_c$ containing $q$.

We can also modify this data structure such that it can answer
$\eps$-incidences queries {\em exactly}. That is, given a query point $q$, it can count (or report) the number of
hyperplanes at distance at most $\eps$ from $q$ and only these hyperplanes.
To do this we  maintain pointers from each hyperplane $h$ in $H_\tau$ to the hyperplanes
in $H_{p(\tau)}$ that
got rounded to $h$.
To answer a query $q$, we find the leaf cell $\tau$ containing $q$ and then we traverse back the pointers of the hyperplanes of
$H_{\tau}$ all the way up the octree to identify the original hyperplanes that were rounded to them.
We then traverse this set of original hypeprlanes and count (or report) those that are at distance at most $\eps$ from $q$.

\end{document}